\documentclass[journal]{IEEEtran}

\usepackage[english]{babel}
\usepackage[utf8x]{inputenc}
\usepackage[T1]{fontenc}
\usepackage[colorinlistoftodos]{todonotes}
\usepackage{cite}

\usepackage{amsmath,amssymb,mathrsfs,bbm,amsthm}
\usepackage{graphicx}
\usepackage[colorinlistoftodos]{todonotes}
\usepackage[colorlinks=true, allcolors=blue]{hyperref}
\usepackage{multirow}
\usepackage[lined,linesnumbered,ruled]{algorithm2e}

\newcommand{\cp}[1]{\ifmmode {\mathcal{#1}}\else ${\mathcal{#1}}$\fi}
	
\newcommand{\bA}{\boldsymbol{A}}
\newcommand{\bB}{\boldsymbol{B}}

\newcommand{\bI}{\boldsymbol{I}}
\newcommand{\bK}{\boldsymbol{K}}
\newcommand{\bM}{\boldsymbol{M}}

\newcommand{\bW}{\boldsymbol{W}}
\newcommand{\bX}{\boldsymbol{X}}
\newcommand{\bY}{\boldsymbol{Y}}

\newcommand{\ba}{\boldsymbol{a}}

\newcommand{\bc}{\boldsymbol{c}}

\newcommand{\bm}{\boldsymbol{m}}

\newcommand{\be}{\boldsymbol{e}}

\newcommand{\by}{\boldsymbol{y}}

\newcommand{\bu}{\boldsymbol{u}}

\newcommand{\bx}{\boldsymbol{x}}

\newcommand{\psh}[2]{\langle{#1},{#2}\rangle_{\cp{H}}}

\newcommand{\mr}[1]{{\widetilde{\bm}_{{#1}}}}

\newcommand{\bxi}{\boldsymbol{\xi}}

\newcommand{\calC}{\mathcal{C}}
\newcommand{\calD}{\mathcal{D}}

\newcommand{\bbeta}{\boldsymbol{\beta}}

\newcommand{\bgamma}{\boldsymbol{\gamma}}

\newcommand{\bmu}{\boldsymbol{\mu}}

\newcommand{\bomega}{\boldsymbol{\omega}}
\newcommand{\bzeta}{\boldsymbol{\zeta}}

\newcommand{\bSigma}{\boldsymbol{\Sigma}}

\newcommand{\cb}[1]{\boldsymbol{#1}}
\newcommand{\tr}{\operatorname{tr}}
\newcommand{\Ex}{\operatorname{E}}

\newcommand{\Hess}{\operatorname{Hess}}
\newcommand{\Homogeneity}{\operatorname{Hom}}

\usepackage{color}  
\def\cred{\textcolor{red}}

\def\cblue{\textcolor{blue}}

\definecolor{darkgreen}{rgb}{0.0, 0.85, 0.0}

\newtheorem{theorem}{Theorem}

\title{A Blind Multiscale Spatial Regularization Framework for Kernel-based Spectral Unmixing}

\author{Ricardo~Augusto~Borsoi,~\IEEEmembership{Student Member,~IEEE,} Tales Imbiriba,~\IEEEmembership{Member,~IEEE,} Jos\'e~Carlos~Moreira~Bermudez,~\IEEEmembership{Senior~Member,~IEEE}, C\'edric Richard,~\IEEEmembership{Senior~Member,~IEEE}
\thanks{This work has been supported by the National Council for Scientific and Technological Development (CNPq) under grants 304250/2017-1, 409044/2018-0, 141271/2017-5 and 204991/2018-8, by the Foundation for Research Support of the State of Rio Grande do Sul (FAPERGS) under grant 19/2551-0001844-4, and by the Brazilian Education Ministry (CAPES) under grant PNPD/1811213.}
\thanks{The authors would like to thank Lucas Drumetz and his collaborators for providing part of the data used in the experimental section of the manuscript.}
\thanks{R.A. Borsoi is with the Department of Electrical Engineering, Federal University of Santa Catarina (DEE--UFSC), Florian\'opolis, SC, Brazil, and with the Lagrange Laboratory (CNRS, OCA), Universit\'e  C\^ote  d'Azur, Nice, France. e-mail: \mbox{raborsoi@gmail.com}.}
\thanks{T. Imbiriba was with DEE--UFSC, Florian\'opolis, SC, Brazil, and is with the ECE department of the Northeastern University, Boston, MA, USA. e-mail: \mbox{talesim@gmail.com}.}
\thanks{J.C.M. Bermudez is with the DEE--UFSC, Florian\'opolis, SC, Brazil, and with the Graduate Program on Electronic Engineering and Computing, Catholic University of Pelotas (UCPel) Pelotas, Brazil. e-mail: \mbox{j.bermudez@ieee.org}.}
\thanks{C. Richard is with the  Lagrange Laboratory (CNRS, OCA), Universit\'e  C\^ote  d'Azur,  Nice, France. e-mail: \mbox{cedric.richard@unice.fr}.}
\thanks{This paper has supplementary downloadable material available at http://ieeexplore.ieee.org., provided by the authors. The material includes more detailed experimental validations. Contact \mbox{raborsoi@gmail.com} for further questions about this work.}
\thanks{Manuscript received Month day, year; revised Month day, year.}
}

\begin{document}
\maketitle

\begin{abstract}

Introducing spatial prior information in hyperspectral imaging (HSI) analysis has led to an overall improvement of the performance of many HSI methods applied for denoising, classification, and unmixing.
Extending such methodologies to nonlinear settings is not always straightforward, specially for unmixing problems where the consideration of spatial relationships between neighboring pixels might comprise intricate interactions between their fractional abundances and nonlinear contributions.
In this paper, we consider a multiscale regularization strategy for nonlinear spectral unmixing with kernels. The proposed methodology splits the unmixing problem into two sub-problems at two different spatial scales: a coarse scale containing low-dimensional structures, and the original fine scale. The coarse spatial domain is defined using superpixels that result from a multiscale transformation. 
Spectral unmixing is then formulated as the solution of quadratically constrained optimization problems, which are solved efficiently by exploring their strong duality and a reformulation of their dual cost functions in the form of root-finding problems.
%
%
Furthermore, we employ a theory-based statistical framework to devise a consistent strategy to estimate all required parameters, including both the regularization parameters of the algorithm and the number of superpixels of the transformation, resulting in a truly blind (from the parameters setting perspective) unmixing method.
Experimental results attest the superior performance of the proposed method when comparing with other,  state-of-the-art, related strategies.




\end{abstract}

\allowdisplaybreaks



\section{Introduction}

Modern remote sensing greatly relies on hyperspectral (HS) image analysis to retrieve information about surface materials in many applications such as agriculture, surveillance and space exploration~\cite{Bioucas-Dias-2013-ID307}. Specifically, reflectance measures can provide detailed information about the spectral signature of pure materials present on the surface of a target scene and their proportion for each pixel. Applications often, but not exclusively, associated with remote sensing trade poor spatial resolution for high spectral resolution due to physical limitations of imaging sensors and to the distance between the sensor and the target scene. Therefore, the measured reflectance of a given pixel is usually a mixture of the pure spectral signatures of materials existing in the corresponding area. 
\emph{Spectral unmixing} (SU) consists of extracting the pure component spectral signatures and their proportions (or abundances) for each pixel. The literature presents many mixing models to explain the observed reflectance as a mathematical function of the pure spectral components. The simplest form of such models is the \emph{linear mixing model} (LMM) which confines the observed reflectance vectors into a convex hull whose extremities are the pure component spectral signatures, therefore, called endmembers.
The LMM is effective in accurately modelling mixtures occurring in scenes where the materials of interest cover a large area with respect to the pixel size~\cite{Dobigeon-2014-ID322}. It however disregards more complex mixing phenomena such as non-linearity~\cite{Dobigeon-2014-ID322,Imbiriba2016_tip} and spectral variability~\cite{somers2011variabilityReview,imbiriba2018glmm,Borsoi_2018_Fusion,borsoi2019deepGun,borsoi2019EMlibManInterpVAE}, which often results in estimation errors being propagated throughout the unmixing process~\cite{Thouvenin_IEEE_TSP_2016_PLMM}.

Nonlinear interactions between materials occur in many scenes where there is complex radiation scattering among several endmembers, such as in some vegetation areas~\cite{Ray1996}. In such situations, nonlinear mixing models must be considered~\cite{Dobigeon-2014-ID322,heylen2014review}. Several nonlinear SU strategies have been proposed in the literature, which can be roughly divided between model-based and model-free methods. Most model-based nonlinear SU algorithms assume that the mixing process that occurs in the scene is known a priori~\cite{mustard1989photometric, Guilfoyle2001,Altmann_tip_2012,heylen2014review,Dobigeon-2014-ID322}. However, real mixing mechanisms can be very complex and prior knowledge about them is seldom available in practice. This led to the consideration of more flexible model-free nonlinear SU, which employ more flexible nonlinear mixing models that are able to represent generic functions. Prominent model-free strategies include the estimation of abundances as posterior class probabilities of a nonlinear classifier~\cite{Li_svm_2005}, the use of graph-based approximate geodesic distances~\cite{Heylen:2011kc,Heylen_2014}, and kernel-based algorithms~\cite{Wu2010,Li2012blind, Chen-2013-ID321, chen2013nonlinear, Altmann-2013-ID311}. Kernel-based methods provide non-parametric representations of functional spaces that are able to model arbitrary nonlinear mixtures~\cite{Wu2010,Li2012blind,heylen2014review,Dobigeon-2014-ID322, Chen-2013-ID321, chen2013nonlinear,ammanouil2016nonlinear}. This flexibility, allied to a good experimental performance has led to the wide application of kernel methods.

Despite the good results obtained with kernel-based unmixing methods~\cite{chen2013nonlinear}, most algorithms fail to explore the high spatial regularity associated to many real world scenes.
This property can be leveraged to improve the conditioning of the unmixing problem. Spatial regularization has already been shown to improve the performance of linear~\cite{shi2014StatialInfoReview,imbiriba2018ULTRA} and sparse~\cite{iordache2012sunsal_TV,feng2016adaptiveRegularizationParameterSparseHU} SU, as well as spectral-variability-aware SU~\cite{drumetz2016blindUnmixingELMM,imbiriba2018ULTRA_V,imbiriba2018glmm,borsoi2019icassp,hong2019augmentedLMMvariability,hong2018SULoRA_lowRankEnbeddingUnmixingVar}, which is closely connected to nonlinear~SU~\cite{drumetz2017relationshipsBilinearELMM}. 
However, spatial information has seldom been enforced in nonlinear unmixing algorithms, partly due to the challenges associated with more complex observation models.
For instance, a spatial clustering approach was used in~\cite{tang2018spatialRegNonlinearUnmixing} to divide the image into different groups of pixels. SU was then performed using the P-linear mixing model in a Bayesian framework with a unique set of regularization parameters for each group.
In~\cite{chen2014nonlinear2} a \emph{Total Variation} (TV) regularization was introduced in a regression-based kernel unmixing method~\cite{chen2013nonlinear}, and a variable splitting approach was then used to solve the resulting optimization problem.

The TV regularization has been widely used in many HS imaging tasks since it promotes smooth image reconstructions while still allowing for sharp discontinuities~\cite{iordache2012sunsal_TV}. However, TV regularization is not the most effective approach to extract spatial information from hyperspectral images. Regularization strategies exploiting nonlocal redundancy in images were recently considered for SU~\cite{wang2017centralizedNonlocalSparseUnmixing,yao2019nonlocalTV_NMF_unmixing}, leading to a better abundance estimation performance at the expense of an increase in computational complexity. Other works exploited the manifold structure in the hyperspectral data by using a graph-based regularization to connect the abundances of pixels that are similar with respect to \mbox{some metrics~\cite{lu2012manifoldRegularizedSparseNMF,ammanouil2015graphLaplacianRegUnmixing}.}

In~\cite{Borsoi_multiscale_lgrs_2018}, a multiscale spatial regularization approach was proposed for sparse spectral unmixing. The multiscale approach led to improved results and smaller computational complexity when compared to TV regularization. The unmixing problem was split into two simpler problems in different image domains defined using a multiscale transformation. This transformation groups image pixels into contiguous regions using (over)-segmentation strategies such as the superpixel decomposition~\cite{Borsoi_multiscale_lgrs_2018}. The multiscale regularization strategy was later extended in~\cite{Borsoi_multiscaleVar_2018} to consider SU accounting for spectral variability.
%
Despite the excellent results obtained with spatial regularization strategies, their performance usually depends on the careful selection of regularization parameters. This is specially important in multiscale strategies, which require a larger number of parameters. Determination of parameters for spatial regularization methods remains a challenging problem, and works applied to HSI are rare~\cite{feng2016adaptiveRegularizationParameterSparseHU,song2016regularizationParamEstimHSdeconvolution}.

In this paper, we propose a new multiscale spatial regularization approach for kernel-based nonlinear unmixing. Building upon the ideas proposed in~\cite{Borsoi_multiscale_lgrs_2018}, we employ a multiscale representation to divide the unmixing problem into two simpler problems in different scales. Though based on the same principle used in~\cite{Borsoi_multiscale_lgrs_2018}, devising kernel-based mixing models in multiple scales is more challenging than in the linear case. Moreover, we address the parameter adjustment problem differently from what has been done in previous multiscale SU formulations in~\cite{Borsoi_multiscale_lgrs_2018,Borsoi_multiscaleVar_2018}. In this work we reformulate the SU problem at multiple scales by statistically characterizing not only the algorithm reconstruction error in both scales, but also the inter-scale interaction between the abundances and the nonlinear mixing contributions across the coarse and fine image domains.
%
%
%
This formulation leads to physically motivated constraints which are leveraged to devise the \emph{Blind Multiscale Unmixing Algorithm for Nonlinear Mixtures} (BMUA-N), in which all the parameters are determined automatically from the observed data. Thus, the proposed strategy benefits from an improved quality without the need for \emph{ad hoc} parameter adjustment such as in TV-based works. 

We formulate the resulting unmixing problem as a sequence of two optimization problems with quadratic equality constraints. These non-convex problems are solved by reformulating the dual problems in the form of low-dimensional root finding problems, which can be solved in very few iterations using a multidimensional bisection algorithm. Moreover, we are able to prove that, under mild conditions, strong duality holds for these optimization problems, which guarantees the optimality of this approach.
%
%
%
%
Simulations with synthetic and real datasets illustrate the effectiveness of the proposed methodology in producing piecewise smooth solutions while preserving sharp discontinuities existing in the image. This leads to more accurate unmixing results when compared to TV-based strategies, with less computational complexity and without the need for \emph{ad hoc} parameter adjustment.


This manuscript is organized as follows. In Section~\ref{sec:kernel} we discuss the main concepts related to regression-based kernel unmixing. In Section~\ref{sec:mscale_nonl} we present the proposed kernel-based multiscale unmixing strategy. The automatic parameter setting methodology is presented in Section~\ref{sec:param_alg} and the solution for the proposed optimization problems is discussed in Section~\ref{sec:prob_opt_solution}. In Section~\ref{sec:superpixel_size_fnder} we propose a method for designing the multiscale transformation to yield spectral homogeneity. Section~\ref{sec:complexity} discusses the computational complexity of the proposed method.
Experimental results are presented and discussed in Section~\ref{sec:results}. They are followed by concluding remarks in Section~\ref{sec:conclusions}.

\section{Kernel-Based Unmixing}\label{sec:kernel}
\subsection{Kernel-based mixture model}


In this section we review the standard kernel-based mixture model introduced in~\cite{chen2013nonlinear} and discuss the main theoretical aspects of kernel machines. As in~\cite{chen2013nonlinear}, we assume that each $L$--band observed pixel $\by_n\in\mathbb{R}^L$ in an HSI can be modeled as a function of the endmember spectra as follows:
\begin{align} \label{eq:nonlin_model_i}
	y_{n,\ell} = \psi_{\ba_n} (\widetilde{\bm}_{\ell}) + e_{n,\ell} \,,\quad
    \ell=1,\ldots,L \,,
\end{align}
where $y_{n,\ell}$ is the $\ell$-th entry of vector $\by_n$, $\widetilde{\bm}_{\ell}\in\mathbb{R}^{1\times P}$ is the $\ell$-th row of the endmember matrix $\bM\in\mathbb{R}^{L\times P}$ with~$P$ spectral signatures of pure materials in the scene. Function $\psi_{\ba_n}$ is an unknown nonlinear function defining the interactions between endmember spectra parameterized by their fractional abundances $\ba_n\in\mathbb{R}^P$.  $e_{n,\ell}$ includes the observation noise and modeling errors.
%
The problem that arises is to find a functional $\psi_{\ba_n}$ that can accurately represent the different and complex types of light-endmember interactions often occurring in real scenes. Since the type of nonlinearity is rarely known in practice, a popular solution is to search for kernel-based smooth function representations whose parameters can be learned directly from the data~\cite{Dobigeon-2014-ID322, Imbiriba2016_tip}.

In~\cite{chen2013nonlinear,chen2014nonlinear2} the authors considered a semi-parametric kernel-based model consisting of a linear trend parameterized by the abundance vector plus an additive nonlinear fluctuation. The model, which allows the quantification of the abundance vectors during the unmixing process, is given by   
\begin{align} \label{eq:nonlin_model_ii}
	\psi_{\ba_n} (\widetilde{\bm}_{\ell}) = \widetilde{\bm}_{\ell} \ba_n 
    + \psi_n(\widetilde{\bm}_{\ell}) \,,
\end{align}
with $\psi_n:\mathbb{R}^P\to\mathbb{R}$ being an arbitrary smooth function belonging to a Reproducing Kernel Hilbert Space (RKHS) denoted by~$\mathcal{H}$ and defined over a nonempty compact set $\mathcal{M}\subset\mathbb{R}^P$.
%
%
%
This assumption allows for the kernel machinery (i.e., via the kernel trick) to obtain accurate solutions to the unmixing problem.



The theory of positive definite kernels emerged from the study of positive definite integral operators~\cite{Mercer1909}, and was further generalized in the study of positive definite matrices~\cite{Moore1916}. 
It has been established that, to every positive definite function $\kappa(\cdot,\cdot):\cp{M}\times \cp{M}\rightarrow\mathbb{R}$, defined over a non-empty compact set $\cp{M}\subset\mathbb{R}^P$, there corresponds one and only one class of real-valued functions on $\cp{M}$ forming a Hilbert space $\cp{H}$ endowed with a uniquely defined inner product $\psh{\cdot}{\cdot}$, and admitting~$\kappa$ as a \emph{reproducing kernel} (r.k.)~\cite{Aronszajn1950}. 
Space $\cp{H}$ is called a RKHS if its evaluation functional $\delta_{\widetilde{\bm}}$ is a linear and continuous (or equivalently bounded) functional for every $\widetilde{\bm}\in\cp{M}$, thus, admitting $\kappa$ as its unique kernel. As a consequence of the Riesz representation theorem~\cite[p. 188]{kreyszig1989introductory}, $\kappa(\cdot,\widetilde{\bm})$ is the representer of evaluation of any functional $\psi\in\cp{H}$, such that the \emph{reproducing property}
\begin{equation}
  \psi(\widetilde{\bm}) = \psh{\psi}{\kappa(\cdot, \widetilde{\bm})}
  \label{eq:reproducingProp}
 \end{equation}
 holds, for all $\psi\in\cp{H}$ and all $\widetilde{\bm}\in \cp{M}$. Furthermore, since $\kappa(\cdot, \widetilde{\bm})\in\cp{H}$ for all $\widetilde{\bm},\widetilde{\bm}'\in\cp{M}$ we also have
 \begin{equation}
 \kappa(\widetilde{\bm},\widetilde{\bm}') = \psh{\kappa(\cdot, \widetilde{\bm})}{\kappa(\cdot, \widetilde{\bm}')} \,.
 \end{equation}
The RKHS $\cp{H}$ is then formed by a class of functions generated by all functions of the form $\psi(\cdot)=\sum_{j}\alpha_j\kappa(\cdot,\widetilde{\bm}_j)$, with norm defined by $\|\psi\|^2_{\cp{H}} = \sum_{i}\sum_{j}\alpha_i\alpha_j\kappa(\widetilde{\bm}_i,\widetilde{\bm}_j)$.

In the context of machine learning, kernel methods are often related with the concept of building a high dimensional feature space $\cp{H}$, and a mapping
\begin{equation}
\begin{split}
	{\cb\Phi}:\,\, \cp{M}&\longrightarrow    \mathcal{H} \\
	 \widetilde{\bm} &\longmapsto   \cb{\Phi}(\widetilde{\bm}) \,,
\end{split}
\end{equation}
with inner product defined as $\kappa(\widetilde{\bm},\widetilde{\bm}')=\psh{\cb{\Phi}(\widetilde{\bm})}{\cb{\Phi}(\widetilde{\bm}')}$. If $\kappa$ is a r.k. of $\cp{H}$, then $\cp{H}$ is a RKHS and also a feature space of $\kappa$ with $\cb{\Phi}(\widetilde{\bm})=\kappa(\cdot, \widetilde{\bm})$. In this case $\cb{\Phi}$ is called the \emph{canonical feature map}~\cite[p.  120]{steinwart2008support}. 
This leads to the so-called ``\emph{kernel trick}'' allowing one to compute inner products of data mapped into higher, or even infinite, dimensional feature spaces by evaluating a real function $\kappa(\widetilde{\bm}_{i},\widetilde{\bm}_{j})$ in the input space.

Although the literature proposed a variety of kernel functions elaborated during the past two decades of intense research activity~\cite{Vapnik1995, ScholkopfBook:2001, Rasmussen-2006-ID292}, in this work we restrain ourselves to the polynomial kernel due to its intimate relation with multiple scattering phenomena known to exist in the interaction between light and the materials in the scene. Thus, the polynomial kernel is given by
\begin{equation} \label{eq:polyKernel}
    \kappa(\widetilde{\bm}_{i},\widetilde{\bm}_{j}) = (\widetilde{\bm}^\top_{i}\widetilde{\bm}_{j} + c)^d \,,
\end{equation}
where $d$ is the polynomial degree and $c\geq 0$ is a real number. 
Due to relevant findings reported in~\cite{Somers:2009p6577} concerning the order of multiple reflection models and the good results obtained in~\cite{chen2014nonlinear2}, in this paper we assume $d=2$ and $c=1$ in all simulations.

\subsection{LS-SVR-based unmixing} \label{sec:LSSVR_unmixing}

In~\cite{chen2013nonlinear} the authors proposed to solve the unmixing problem accounting for the model in~\eqref{eq:nonlin_model_i}--\eqref{eq:nonlin_model_ii} by considering a multi-kernel generalization of standard least-squares support vector regression (LS-SVR) methods~\cite{Suykens2002}. 
The resulting optimization problem is given by
\begin{align}
 \label{eq:optproblem_lssvrU}
    & (\hat{\ba}_n,\hat{\psi}_n) =  \,\,  \mathop{\arg\min}_{\ba_n,\,\psi_n} \,\, \frac{1}{2}\Big(\|\ba_n\|^2+\|\psi_n\|_\cp{H}^2+\frac{1}{\mu}\|\bxi_n\|_2^2\Big)
    \\
    &\text{subject to}\,\,\, \ba_n \geq \cb{0} \,, \, \cb{1}^\top \ba_n = 1 \,, 
    \nonumber \\ \nonumber 
    &\qquad\qquad\,\,\, \xi_{n,\ell} = y_{n,\ell} - \,\ba_n^\top \mr{\ell} - \psi_n(\mr{\ell}), \,\, \ell=1,\ldots,L \,,
\end{align}
where $\hat{\ba}_n$ and $\hat{\psi}_n$ are the estimated abundance vector and nonlinear function for the $n-$th pixel.





Problem~\eqref{eq:optproblem_lssvrU} is solved using standard dual formulation based on the Lagrangian~\cite{chen2013nonlinear}. 
Although problem~\eqref{eq:optproblem_lssvrU} presents an effective way of modeling both the linear trend and the nonlinear mixing occurring in a given pixel, it fails to impose any smooth structure over the abundance estimation within neighboring pixels. 

Standard regularization approaches such as the TV have considered an additional term to the cost function in problem~\eqref{eq:optproblem_lssvrU} that penalizes spatial discontinuities in the abundance maps~\cite{chen2014nonlinear2}. However, as discussed in the introduction, this strategy is not the most effective in exploring spatial information contained in the image. Besides, it introduces an additional parameter that must be carefully tuned in order to achieve good performance.

In the following, we present a multiscale formulation for the nonlinear mixing model in~\eqref{eq:nonlin_model_i}--\eqref{eq:nonlin_model_ii} that enables us to better exploit the spatial regularity in HSIs, leading to improved results when compared to the TV-based strategy.
%
%
Compared to linear unmixing, the application of a multiscale formulation to kernel-based nonlinear unmixing leads to specific challenges that need to be addressed. This is specially true in the present work, as quadratic equality constraints need to be reformulated to allow for the automatic determination of the parameters, and thus for a blind algorithm as detailed in Section~\ref{sec:param_alg}.

\section{A multiscale nonlinear mixing model} \label{sec:mscale_nonl}

%

Traditional regularization approaches (e.g., Tikhonov or TV) introduce spatial regularity by promoting similarity between abundances at fixed spatial neighborhoods. Recently, more flexible approaches emerged exploring irregular and data-dependent image structures that can generally be described under a graph or manifold regularization framework~\cite{lu2012manifoldRegularizedSparseNMF,shuman2013signalProcessingGraphsReview,ammanouil2015graphLaplacianRegUnmixing}. In this case, a graph is first constructed to represent the similarity between pairs of pixels in the HSI using distance metrics that can be either defined explicitly~\cite{stevens2017graphConstructionHSIs} or derived indirectly using image (over)-segmentation methods such as superpixels, ultrametric contour maps (UCM), or binary partition trees (BPT)~\cite{achanta2012slicPAMI,arbelaez2006ultrametricContourMaps,veganzones2014hyperspectralSegmentationBPT}. Afterwards, abundances corresponding to pixels that are similar according to the graph can be constrained to have similar proportions in the SU problem, which preserves the geometric structure found in the HSI~\cite{lu2012manifoldRegularizedSparseNMF,ammanouil2015graphLaplacianRegUnmixing,wang2017superpixelSparseNMF}.
Despite providing a significant amount of flexibility, graph regularization can still lead to computationally costly SU problems. More importantly though, this approach does not provide a clear insight about the multiscale abundance interactions in a way that could motivate the design of the algorithms, which makes the selection of the regularization parameters difficult in practical scenarios.

Motivated by the results in~\cite{Borsoi_multiscale_lgrs_2018,Borsoi_multiscaleVar_2018}, we propose to introduce spatial information into SU by representing this problem separately in two spatial scales, which significantly reduces the computational complexity of SU. Moreover, this also makes the inter-scale interaction between the abundances explicit in the resulting optimization problem, which is essential for the theoretically principled parameter design strategy presented in Section~\ref{sec:param_alg}.
Specifically, we divide the SU problem into two consecutive steps. First, we represent the nonlinear mixing process in an approximation (coarse) spatial scale~($\mathcal{C}$) which preserves relevant inter-pixel spatial contextual information. Pixels in the coarse spatial scale can be unmixed independently from each other. The recovered coarse abundance maps are then mapped back to the original image domain ($\mathcal{D}$) and used as prior information to regularize the second unmixing process applied to the original image to promote spatial dependency between neighboring pixels.

\subsection{Unmixing in the coarse scale}
Denote the HSI and the abundance map for all pixels by $\bY=[\by_1,\ldots,\by_N]$ and $\bA=[\ba_1,\ldots,\ba_N]$, respectively.
We consider a dimensionality reduction transformation~$\bW \in \mathbb{R}^{N\times K}$, $K < N$, constructed based on relevant contextual inter-pixel information present in the observed image $\bY$, that maps both the HSI and the abundance map to the approximation domain. The transformed matrices are given by
\begin{align} \label{eq:decomposition_calC_i}
	\bY_{\!\calC} = \bY\bW \,, 
    \quad \bA_{\calC} = \bA \bW \,,
\end{align}
where~$\bY_{\!\calC}=[\by_{\calC_1},\ldots,\by_{\calC_K}] \in \mathbb{R}^{L\times K}$ and~$\bA_{\calC}=[\ba_{\calC_1},\ldots,\ba_{\calC_K}]  \in \mathbb{R}^{P\times K}$ are, respectively, the HSI and the abundance matrix in the coarse approximation scale.

Various methods can be used to construct the transformation~$\bW$. Specifically, $\bW$ must group pixels that are spatially adjacent and spectrally similar and must respect image borders by not grouping pixels corresponding to different image structures or features.
Following the same approach as in~\cite{Borsoi_multiscale_lgrs_2018,Borsoi_multiscaleVar_2018}, we consider the superpixel decomposition of the image $\bY$ for the transformation $\bW$. Besides satisfying the criteria outlined above, multiscale decompositions based on superpixel algorithms have shown excellent performance in SU considering both sparsity~\cite{Borsoi_multiscale_lgrs_2018} and variability of the endmembers~\cite{Borsoi_multiscaleVar_2018}.
Superpixel algorithms group image pixels into different spatially compact neighborhoods with similar spectral information~\cite{achanta2012slicPAMI}, decomposing the image into a set of contiguous homogeneous regions whose size and regularity are controlled by adjusting a set of parameters. Furthermore, the superpixel decomposition can be computed very efficiently by employing low-cost algorithms such as the SLIC~\cite{achanta2012slicPAMI}. 
The transformation $\bW$ is thus constructed such that $\bY\bW$ computes the superpixel decomposition of the image $\bY$, and returns the averages of all pixels inside each superpixel region.
%
The effect of the transformation $\bW$ on the Cuprite HSI is illustrated in Figure~\ref{fig:img_ex_sppx3}.
Note that, besides the superpixel decomposition, other image (over)-segmentation strategies can also be used to construct~$\bW$, such as e.g., UCM or BPT~\cite{arbelaez2006ultrametricContourMaps,veganzones2014hyperspectralSegmentationBPT}, which can also provide hierarchical (multiscale) representations of HSIs.


\begin{figure}
    \centering
    \includegraphics[width=0.7\linewidth]{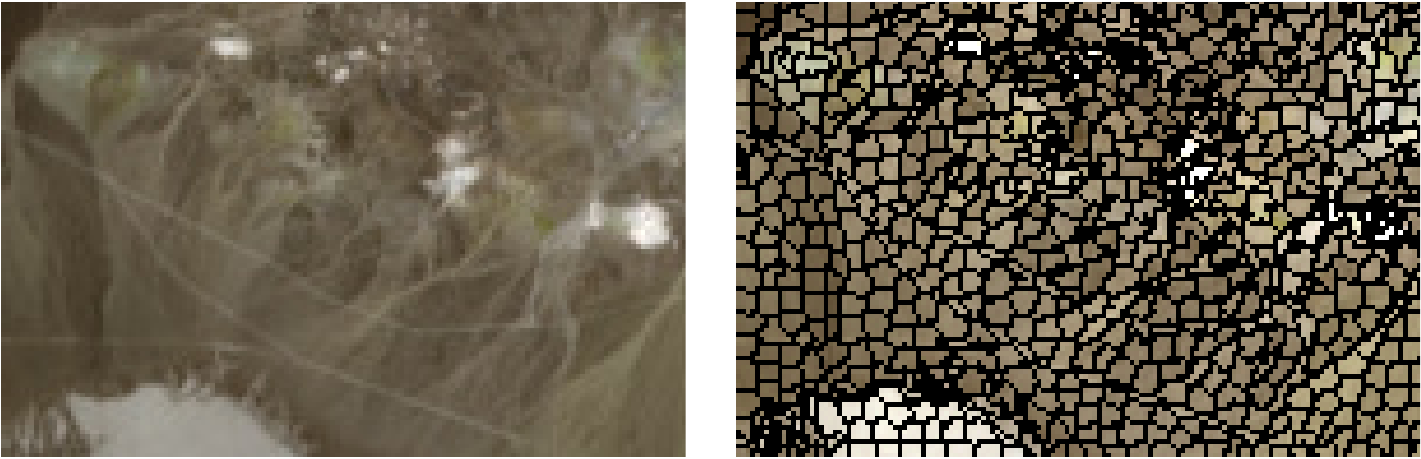}
    \vspace{-0.2cm}
    \caption{Cuprite image (left) and its superpixel decomposition (right).}
    \label{fig:img_ex_sppx3}
\end{figure}

Considering the nonlinear observation model~\eqref{eq:nonlin_model_i}, the transformed image in~\eqref{eq:decomposition_calC_i} leads to an equivalent mixing model in the coarse spatial domain, which is given by
\begin{align} \label{eq:nonl_model_coarse_scale}
	y_{\mathcal{C}_{i,\ell}}
    & = \frac{1}{|\mathcal{N}_i|}
    \sum_{n\in\mathcal{N}_i} \big(\widetilde{\bm}_{\ell} \ba_n 
    + \psi_n (\widetilde{\bm}_{\ell}) + e_{n,\ell} \big)
    \nonumber \\
    & = \widetilde{\bm}_{\ell}  \ba_{\mathcal{C}_i} + \psi_{\mathcal{C}_{i}}(\widetilde{\bm}_{\ell})
	+ \frac{1}{|\mathcal{N}_i|}\sum_{n\in\mathcal{N}_i} e_{n,\ell} \,,
\end{align}
where $y_{\mathcal{C}_{i,\ell}}$ is the $\ell$-th entry of $\by_{\mathcal{C}_{i}}$, $\mathcal{N}_i$ is the set of indexes of the pixels contained within the $i$-th superpixel, $|\cdot|$ denotes the cardinality of a set, and $\ba_{\mathcal{C}_i}$ and $\psi_{\mathcal{C}_{i}}$, given by
\begin{align} \notag
	\ba_{\mathcal{C}_i} {}={} \frac{1}{|\mathcal{N}_i|} \sum_{n\in \mathcal{N}_i} \ba_n
    \,, \quad
    \psi_{\mathcal{C}_{i}}(\widetilde{\bm}_{\ell}) = \frac{1}{|\mathcal{N}_i|} \sum_{n\in \mathcal{N}_i} \psi_n (\widetilde{\bm}_{\ell}) \,,
\end{align}
for $i=1,\ldots,K$, denote the fractional abundances and nonlinear contributions at the coarse spatial scale.

Following the observation model~\eqref{eq:nonl_model_coarse_scale}, the SU problem at the coarse spatial scale can be formulated using the LS-SVR framework presented in Section~\ref{sec:LSSVR_unmixing}, which leads to the following optimization problem:
\begin{align} \label{eq:khype_opt_coarse_constr}
    & \{\hat{\ba}_{\mathcal{C}_i},\hat{\psi}_{\mathcal{C}_i}\} =  \mathop{\arg\min}_{\{\ba_{\mathcal{C}_i},\psi_{\mathcal{C}_i},\bxi_{\mathcal{C}_i}\}} \,\,
    \frac{1}{2} \sum_{i=1}^K 
    \Big( \|\psi_{\mathcal{C}_i}\|_{\mathcal{H}}^2
    + \|\ba_{\mathcal{C}_i}\|_2^2 \Big)
    \\
    & \text{subject to }\, \ba_{\mathcal{C}_i} \geq \cb{0} \,, \,
    \cb{1}^\top\ba_{\mathcal{C}_i} = 1 \,,
    \,\,\, i=1,\ldots,K \,,
    \nonumber \\ & \hspace{1.6cm}
    \bxi_{\mathcal{C}_i} {}={} \by_{\mathcal{C}_i} - \bM \ba_{\mathcal{C}_i} - \psi_{\mathcal{C}_i}(\bM)
    ,\,\,\,\, i=1,\ldots,K \,,
    \nonumber \\
    \nonumber
    & \hspace{1.5cm} \frac{1}{K}\sum_{i=1}^K\|\bxi_{\mathcal{C}_i}\|_2^2 {}={} C_0 \,,
\end{align}
where $\hat{\ba}_{\mathcal{C}_i}$ and $\hat{\psi}_{\mathcal{C}_i}$ are the estimated abundance vector and nonlinear function for the $i-$th superpixel and $\psi_{\mathcal{C}_{i}}(\bM)=[\psi_{\mathcal{C}_{i}}(\widetilde{\bm}_{1}),\ldots,\psi_{\mathcal{C}_{i}}(\widetilde{\bm}_{L})]^\top$.
Parameter $C_0$ is a positive constant that constrains the reconstruction error of the algorithm, and operates in an analogous manner to a regularization. 
%
%
%
%
Differently from~\eqref{eq:optproblem_lssvrU}, in~\eqref{eq:khype_opt_coarse_constr} we choose to limit the reconstruction error using an equality constraint instead of directly adding $\|\bxi_{\mathcal{C}_i}\|_2^2$, $i=1,\ldots,K$ as a regularizer in the cost function. The rationale is that, unlike the regularization parameter $\mu$ in~\eqref{eq:optproblem_lssvrU}, the constant $C_0$ has a clear physical interpretation. This characteristic is exploited in the next section where we present a proper methodology for automatically setting~$C_0$.

\subsection{Unmixing in the image domain}
The abundance maps estimated at the coarse spatial scale, denoted by $\bA_{\mathcal{C}}=[\ba_{\mathcal{C}_1},\ldots,\ba_{\mathcal{C}_K}]$, can be used to regularize the original unmixing problem. To this end, we convert the abundance map from the coarse approximation domain $\mathcal{C}$ back to the original image domain $\mathcal{D}$ as
\begin{align} \label{eq:mscale_conjugate_transf}
	\,\widehat{\!\bA}_{\mathcal{D}}=\,\widehat{\!\bA}_{\mathcal{C}}\bW^* \,.
\end{align}
Matrix $\bW^\ast\in\mathbb{R}^{K\times N}$ is a conjugate transformation to $\bW$, and takes the image from the coarse domain $\calC$ back to the original (uniform) image domain. This is performed by attributing the value $\widehat{\ba}_{\mathcal{C}_i}$ to all pixels in $\,\widehat{\!\bA}_{\mathcal{D}}$ that lie within the $i$-th superpixel. Thus, $\,\widehat{\!\bA}_{\mathcal{D}}$ can be viewed as a coarse version of~$\,\widehat{\!\bA}$ in the original image domain. 
After computing $\,\widehat{\!\bA}_{\mathcal{D}}$ using~\eqref{eq:mscale_conjugate_transf}, the SU problem for all pixels is given by
\begin{align} \label{eq:khype_opt_finescale_constr}
	&  \{\hat{\ba}_{n},\hat{\psi}_{n}\} =  \, \mathop{\arg\min}_{\{\ba_{n},\psi_{n},\bxi_n\}} \,\,\, 
    \frac{1}{2} \sum_{n=1}^N \|\psi_{n}\|_{\mathcal{H}}^2
    \\
    & \text{subject to }\, \ba_{n} \geq \cb{0}, \,\, \cb{1}^\top\ba_{n} = 1 \,,
    \,\,\, n=1,\ldots,N \,,
    \nonumber \\ & \hspace{1.6cm}
    \bxi_{n} {}={} \by_{n} - \bM \ba_{n} - \psi_{n}(\bM) \,,
    \,\,\, n=1,\ldots,N \,,
    \nonumber \\
    & \hspace{1.5cm} \frac{1}{N}\sum_{n=1}^N\|\bxi_n\|_2^2 {}={} C_1 \,,
	\,\,\,
    \frac{1}{N}\sum_{n=1}^N
    \|\ba_n-\widehat{\ba}_{\mathcal{D}_n}\|_2^2 {}={} C_2 \,,
    \nonumber 
\end{align}
where $\hat{\ba}_n$ and $\hat{\psi}_n$ are the estimated abundance vector and nonlinear function for the $n-$th pixel in the original image domain, $\psi_{n}(\bM)=[\psi_{n}(\widetilde{\bm}_{1}),\ldots,\psi_{n}(\widetilde{\bm}_{L})]^\top$, and $C_1$, $C_2$ are positive constants that constrain the reconstruction error and the abundance variability across scales.
%
Again, as in~\eqref{eq:khype_opt_coarse_constr}, we use equality constraints instead of additive penalty terms in the cost function due to the easier interpretation of $C_1$ and $C_2$ when compared to regularization parameters. This improved interpretability is exploited in the next section to provide a methodology for automatically adjusting these constants. We note that the nonlinear equality constraints make the optimization problems \eqref{eq:khype_opt_coarse_constr} and~\eqref{eq:khype_opt_finescale_constr} non-convex. An efficient algorithm will be proposed in Section~\ref{sec:prob_opt_solution} to address this issue.



\section{Determining the Regularization Constants}
\label{sec:param_alg}

A significant challenge in regularized unmixing algorithms consists in determining the regularization constants. The proposed formulation requires the selection of $C_0$, $C_1$ and $C_2$ in problems~\eqref{eq:khype_opt_coarse_constr} and~\eqref{eq:khype_opt_finescale_constr}. 
%
Although most works assume that such constants can be determined empirically, several frameworks have been proposed to select regularization parameters in (ill-posed) inverse problems, such as the generalized cross validation~\cite{golub1979generalizedCrossValidationParameterEst}, risk estimation for reconstruction error minimization~\cite{deledalle2014parameterEstimationSUGAR,ammanouil2018adapt_parameterEstimation} or the L-curve~\cite{belge2002multipleParameterLcurve}. However, determining these parameters blindly for a given problem is usually difficult and computationally intensive. This motivates the consideration of external \emph{a priori} information whenever possible, which allows the application of conceptually simpler and more reliable solutions.

Classical approaches in this latter category are the Chi-squared method~\cite{hall1987chiSquaredParameterEstimationGCV,thompson1991studyParameterSelectionRegularizationPAMI,galatsanos1992regularizationParamChoiceImageRestoration} and the closely related discrepancy principle~\cite{engl1996regularizationInverseProbs}. The underlying idea behind these methods is to evaluate the statistical properties of the data reconstruction term in the cost function when the estimated solution is equal to the desired (i.e., true) parameters. For an observation model such as~\eqref{eq:nonlin_model_i}, evaluating the true parameters in the cost function of~\eqref{eq:optproblem_lssvrU} would result in the residuals $\bxi_n$ being equal to the observation noise $\be_n$ -- which, under a Gaussianity hypothesis, makes the reconstruction error Chi-squared-distributed, giving the technique its name. 
The regularization constant is then selected so that the estimated solution yields a residual with the same statistical properties of the observation noise, which are assumed to be known~\cite{hall1987chiSquaredParameterEstimationGCV,thompson1991studyParameterSelectionRegularizationPAMI,galatsanos1992regularizationParamChoiceImageRestoration}. This can be performed using equality constrained optimization problems such as in~\eqref{eq:khype_opt_coarse_constr} and~\eqref{eq:khype_opt_finescale_constr}~\cite{hunt1973parameterEstimationChiSquaredEqualityConstraint}, which immediately translates the choice of constants $C_1$, $C_2$ and $C_3$ as the problem of determining the statistical properties of the corresponding equality terms.

In this section, we will extend this strategy to the present SU problem by determining the statistical averages of the reconstruction error in both scales and the inter-scale abundance variation. 
Although in principle this requires knowledge about the true abundance solutions we want to estimate (which historically limited the applicability of the Chi-square technique to more complex problems like~\eqref{eq:khype_opt_finescale_constr}), we will be able to provide a theoretically sound and yet simple strategy for its reliable estimation.

More precisely, constant $C_0$ in~\eqref{eq:khype_opt_coarse_constr} reflects the average noise power in the coarse image scale. Constants $C_1$ and $C_2$ in~\eqref{eq:khype_opt_finescale_constr} reflect, respectively, the average noise power in the detail scale and the average energy of the differences between the fractional abundances and their estimates in the coarse domain. Under the Chi-squared framework, these constants are defined in terms of statistical means as
%
\begin{align} 
	C_0 & {}={} \Ex\bigg\{
    \frac{1}{K}\sum_{i=1}^K\|\by_{\mathcal{C}_i}-\bM\ba_{\mathcal{C}_i}
    -\psi_{\mathcal{C}_i}(\bM)\|_2^2 \bigg\}
    \nonumber \\
    & {}={} \Ex\bigg\{\frac{1}{K}\sum_{i=1}^K \|\be_{\mathcal{C}_i}\|_2^2 \bigg\} \,,
    \label{eq:constants_expectations0a}
    \\
    C_1 & {}={} \Ex\bigg\{\frac{1}{N}\sum_{n=1}^N\|\by_n-\bM\ba_n
    -\psi_n(\bM)\|_2^2 \bigg\}
    \nonumber \\
    & {}={} \Ex\bigg\{\frac{1}{N}\sum_{n=1}^N \|\be_{n}\|_2^2 \bigg\} \,,
    \label{eq:constants_expectations0b}
    \\
    C_2 & {}={} \Ex\bigg\{\frac{1}{N}\sum_{n=1}^N
    \|\ba_n-{\ba}_{\mathcal{D}_n}\|_2^2 \bigg\}  \,,
    \label{eq:constants_expectations0c}
\end{align}
where $\Ex\{\cdot\}$ is the expected value operator with respect to the true distribution of the parameters within the brackets, and $\be_n$ and $\be_{\calC_i}$ are the errors related to the fine and coarse domains of the $n$-th pixel and $i$-th superpixel, respectively. In this case, the quadratic constraints in problems~\eqref{eq:khype_opt_coarse_constr} and~\eqref{eq:khype_opt_finescale_constr} can be seen as approximations to the above expected values.

Constants $C_0$ and $C_1$ depend directly on the noise level and modeling errors represented by $\be_n$. We write
\begin{equation}
	\be_n = \be_{0,n} + \be_{\psi,n} \,,
\end{equation}
where $\be_{0,n}$ is a Gaussian noise term with zero mean and covariance $\bSigma_{\be}$, and $\be_{\psi,n}$ represents modeling errors. Before proceeding, we make the following assumptions:
\begin{enumerate}
	\item[\textsf{A1})] The additive noise is spatially uncorrelated, i.e., $\Ex\{\be_{0,n}^\top\be_{0,m}\}=0$, $\forall\,n\neq m$. 
    
    \item[\textsf{A2})] The noise and modeling errors $\be_{0,n}$ and $\be_{\psi,n}$ are uncorrelated. 
    
    \item[\textsf{A3})] The modeling errors $\be_{\psi,n}$ are assumed to be spatially correlated and approximately constant within each superpixel, that is:
  \begin{align} \label{eq:mdl_err_sppx_appr}
      \be_{\psi,n} \approx \be_{\psi,m} \,, \,\,\, 
      \forall\, m,n\in \mathcal{N}_i,\,\,\, i=1,\ldots,K
      \,.
  \end{align}
    This hypothesis is motivated by the spatial smoothness of both the abundances and the nonlinear contributions in the mixing model~\cite{altmann2015spatialBayesianNonlinearUnmixing,ammanouil2016spatialNonlinearUnmixing}.

	\item[\textsf{A4})] The expected value of the modeling error's norm is the same for all pixels, and is represented as 
    \begin{align}
    \Ex\big\{\|\be_{\psi,n}\|_2^2 \big\}
    = {\sigma_{\be,{\psi}}^2} \,,
    \quad n=1,\ldots,N
    \,.
    \end{align}
    
    \item[\textsf{A5})] The noise covariance matrix is the same for all image pixels, i.e., 
    \begin{align}
    \begin{split}
        \Ex\big\{\be_{0,i}\be_{0,i}^\top\big\} & = \Ex\big\{\be_{0,j}\be_{0,j}^\top\big\} \,, \,\, 1 \leq i,j \leq N 
        \\
        & = \bSigma_{\be} \,.
    \end{split} 
    \end{align}
    
    \item[\textsf{A6})] The vectors $\ba_n-\ba_{\mathcal{D}_n}$, $\bM^{\dagger}(\psi_n(\bM)  - \psi_{\mathcal{C}_n}(\bM)\big)$ and $\bM^{\dagger}\big(\be_n-\frac{1}{|\mathcal{N}_n|}\sum_{i\in\mathcal{N}_n}\be_i)$ are mutually uncorrelated and zero-mean.
    
\end{enumerate}
We denote the average size of each superpixel by $S=N/K$.

In the following, we will evaluate the expectations in~\eqref{eq:constants_expectations0a}--\eqref{eq:constants_expectations0c} in order to provide well-founded means to select the constants $C_0$, $C_1$ and $C_2$.

\subsection{Determining the constant $C_1$}

Using hypothesis \textsf{A1}, \textsf{A2}, \textsf{A4} and \textsf{A5}, constant $C_1$ can be computed as:
\begin{align} \label{eq:constant_C1}
	C_1 & {}={} \Ex\bigg\{\frac{1}{N}\sum_{n=1}^N \|\be_{0,n} + \be_{\psi,n}\|_2^2 \bigg\}
    \nonumber \\ & 
    {}\stackrel{\textsf{{\tiny (A2)}}}{=}{} \frac{1}{N}\sum_{n=1}^N \Big( \Ex\big\{\|\be_{0,n}\|_2^2\big\} 
    + \Ex\big\{\|\be_{\psi,n}\|_2^2 \big\} \Big)
    \nonumber \\ & 
    \hspace{-0.1cm} {}\stackrel{\textsf{{\tiny (A4,A5)}}}{=}{}  \tr\{\bSigma_e\} + {\sigma_{\be,{\psi}}^2}
    \,.
\end{align}




\subsection{Determining the constant $C_0$}

Constant $C_0$ can be derived in a similar way by assuming \textsf{A1}--\textsf{A4}:
\begin{align}
	C_0 & = \Ex\bigg\{\frac{1}{K}\sum_{i=1}^K \|\be_{\mathcal{C}_i}\|_2^2 \bigg\}
    %
    %
    \nonumber \\ & 
    = \Ex\bigg\{\frac{1}{K}\sum_{i=1}^K \Big\|\frac{1}{|\mathcal{N}_i|}\sum_{n\in\mathcal{N}_i}\big(\be_{0,n} + \be_{\psi,n}\big)\Big\|_2^2 \bigg\}
    \nonumber \\ & 
    = \Ex\bigg\{\frac{1}{K}\sum_{i=1}^K \frac{1}{|\mathcal{N}_i|^2} \Big(
    \Big\|\sum_{n\in\mathcal{N}_i}\be_{0,n}\Big\|_2^2
    + \Big\|\sum_{n\in\mathcal{N}_i}\be_{\psi,n}\Big\|_2^2
    \nonumber \\ & \hspace{10ex}
    + 2\sum_{n\in\mathcal{N}_i}\sum_{m\in\mathcal{N}_i} 
    \langle\be_{0,n},\be_{\psi,m}\rangle \Big)\bigg\}
    \nonumber \\ & 
    {}\stackrel{\textsf{{\tiny (A2)}}}{=}{} \frac{1}{K}\sum_{i=1}^K \frac{1}{|\mathcal{N}_i|^2}
    \sum_{n\in\mathcal{N}_i}\Ex\big\{\|\be_{0,n}\|_2^2\big\}
    \nonumber \\ & \hspace{10ex}
    + \frac{1}{K}\sum_{i=1}^K \frac{1}{|\mathcal{N}_i|^2} 
    \Ex\bigg\{\Big\|\sum_{n\in\mathcal{N}_i}\be_{\psi,n}\Big\|_2^2\bigg\} \,.
\end{align}
The modeling errors are not uncorrelated and zero mean in each superpixel.
By approximating $K^{-1}\sum_{i=1}^K |\mathcal{N}_i|^{-1}\simeq S^{-1}$ and using hypothesis \textsf{A3}, $C_0$ can be approximated as
\begin{align} \label{eq:constant_C0}
    C_0 & \simeq \frac{1}{S} \tr\{\bSigma_{\be}\}
    + \frac{1}{K}\sum_{i=1}^K \frac{1}{|\mathcal{N}_i|^2} 
    \sum_{n\in\mathcal{N}_i} 
    \Ex\Big\{|\mathcal{N}_i|\,\big\|\be_{\psi,n}\big\|_2^2\Big\}
    \nonumber \\ &
    \simeq \frac{1}{S} \tr\{\bSigma_{\be}\}
    + {\sigma_{\be,{\psi}}^2}
    \,.
\end{align}

Note that this shows that the modeling errors are more significant relative to the noise at the coarse spatial scale. This is because the contribution of the noise is reduced by a factor of $S$ whereas the modeling errors retain the same energy.


\subsection{Determining the constant $C_2$}

The constant $C_2$ is slightly more challenging to compute than the previous ones. We denote the left pseudo-inverse of $\bM$ as $\bM^{\dagger}\in\mathbb{R}^{P\times L}$, that is, $\bM^{\dagger}\bM=\bI_{\!P}$. Then, we can consider the following quantity:
\begin{align} \label{eq:expected_constr_aad_1}
	& \Ex\bigg\{\frac{1}{N}\sum_{n=1}^N
    \|\bM^{\dagger}(\by_n-\by_{\mathcal{D}_n})\|^2 \bigg\}
    \nonumber \\ &
    {}={} \Ex\bigg\{\frac{1}{N}\sum_{n=1}^N
    \Big\|\bM^{\dagger}\big( \bM\ba_n + \be_n + \psi_n(\bM) 
    \nonumber \\ & \hspace{1.4cm}
    - \bM\ba_{\mathcal{D}_n} - \psi_{\mathcal{C}_n^*}(\bM)-\frac{1}{|\mathcal{N}_n|}\sum_{i\in\mathcal{N}_n}\be_i \big)\Big\|^2
    \bigg\} \,.
\end{align}

The expectation on the right hand side of~\eqref{eq:expected_constr_aad_1} can be computed by using assumption \textsf{A6}. This assumption is reasonable since each of the terms in \textsf{A6} comprises fluctuations between the coarse and fine spatial scales, which can be expected to be of zero mean.
Thus, by noting that
\begin{align}
	\Ex\big\{\|\bM^{\dagger}(\bM\ba_n -\bM\ba_{\mathcal{D}_n})\|^2 \big\}
    {}={} 
    \Ex\big\{\|\ba_n -\ba_{\mathcal{D}_n}\|^2 \big\}
\end{align}
and using \textsf{A6}, we can write~\eqref{eq:expected_constr_aad_1} as
\begin{align} \label{eq:expected_constr_aad_2}
	& \hspace{-1.5ex}    \Ex\bigg\{\frac{1}{N}\sum_{n=1}^N
    \|\bM^{\dagger}(\by_n-\by_{\mathcal{D}_n})\|^2 \bigg\} = 
    \Ex\bigg\{\frac{1}{N}\sum_{n=1}^N\|\ba_n -\ba_{\mathcal{D}_n}\|^2 \bigg\}
    \nonumber \\ & \hspace{1.3cm}
    + \Ex\bigg\{\frac{1}{N}\sum_{n=1}^N\|\bM^{\dagger}(\psi_n(\bM)  - \psi_{\mathcal{C}_n}(\bM)\big)\|^2 \bigg\}
    \nonumber \\ &  \hspace{1.3cm}
    + \Ex\bigg\{\frac{1}{N}\sum_{n=1}^N\|\bM^{\dagger}\big(\be_n-\frac{1}{|\mathcal{N}_n|}\sum_{i\in\mathcal{N}_n}\be_i)\|^2 \bigg\} \,.
\end{align}

In the following, we will expand the last term on the right hand side of \eqref{eq:expected_constr_aad_2}. Using  \textsf{A1}--\textsf{A4}, the summand for for each pixel can be written as
\begin{align} \label{eq:regDerA_noise_term2}
	& \Ex\Big\|\bM^{\dagger}\Big( \be_n-\frac{1}{|\mathcal{N}_n|}\sum_{i\in\mathcal{N}_n}\be_i \Big)\Big\|^2
    \nonumber \\ & 
    = \Ex\Big\|\bM^{\dagger}\Big(\be_{0,n} + \be_{\psi,n}-\frac{1}{|\mathcal{N}_n|}\sum_{i\in\mathcal{N}_n} \big(\be_{0,i} + \be_{\psi,i}\big) \Big)\Big\|^2
    \nonumber \\
	& \stackrel{\textsf{{\tiny (A2)}}}{=} \Ex\Big\|\bM^{\dagger}\Big(\be_{0,n}-\frac{1}{|\mathcal{N}_n|}\sum_{i\in\mathcal{N}_n}\be_{0,i} \Big)\Big\|^2
    \nonumber \\ & \hspace{10ex}
    + \Ex\Big\|\bM^{\dagger}\Big(
    \be_{\psi,n}-\frac{1}{|\mathcal{N}_n|}\sum_{i\in\mathcal{N}_n}\be_{\psi,i}\Big)\Big\|^2 \,.
\end{align}
Using hypothesis \textsf{A3} and equation \eqref{eq:mdl_err_sppx_appr}, the second term of \eqref{eq:regDerA_noise_term2} can be approximated as
\begin{align}
	\Ex\Big\|\bM^{\dagger}\Big(
    \be_{\psi,n}-\frac{1}{|\mathcal{N}_n|}\sum_{i\in\mathcal{N}_n}\be_{\psi,i}\Big)\Big\|^2 \approx 0 \,.
\end{align}
This can be intuitively reasoned by accounting for the spatial correlation of the modeling errors, where $\be_{\psi,n}$ in each pixel is very similar to the average of the modeling errors in its respective superpixel. This leads to
\begin{align}
	& \Ex\Big\|\bM^{\dagger}\big(\be_n-\frac{1}{|\mathcal{N}_n|}\sum_{i\in\mathcal{N}_n}\be_i)\Big\|^2
    \nonumber \\ &
    \simeq \Ex\Big\|\bM^{\dagger}\Big(\be_{0,n}-\frac{1}{|\mathcal{N}_n|}\sum_{i\in\mathcal{N}_n}\be_{0,i} \Big)\Big\|^2
    \nonumber \\ &
    {}={} \Ex\Big\|\bM^{\dagger}\be_{0,n} \big(1-\frac{1}{|\mathcal{N}_n|}\big) - \bM^{\dagger}\frac{1}{|\mathcal{N}_n|}\sum_{i\in\mathcal{N}_n\setminus\{n\}}\be_{0,i}\Big\|^2
    \nonumber \\ &
    \stackrel{\textsf{{\tiny (A1)}}}{=} \Ex\Big\|\bM^{\dagger}\be_{0,n}\big(1-\frac{1}{|\mathcal{N}_n|}\big)\Big\|^2 
    + \Ex\Big\|\bM^{\dagger}\frac{1}{|\mathcal{N}_n|}\sum_{i\in\mathcal{N}_n\setminus\{n\}}\be_{0,i} \Big\|^2
    \nonumber \\ &
    \stackrel{\textsf{{\tiny (A1)}}}{=} \Ex\Big\|\bM^{\dagger}\be_{0,n}\big(1-\frac{1}{|\mathcal{N}_n|}\big)\Big\|^2 
    + \frac{1}{|\mathcal{N}_n|^2} \sum_{i\in\mathcal{N}_n\setminus\{n\}} 
    \Ex\|\bM^{\dagger}\be_{0,i} \|^2
    \nonumber \\ &
    \stackrel{\textsf{{\tiny (A5)}}}{=} \Ex\|\bM^{\dagger}\be_{0,n}\|^2 \big(1-\frac{1}{|\mathcal{N}_n|}\big)^2
    + \frac{|\mathcal{N}_n|-1}{|\mathcal{N}_n|^2}\Ex\|\bM^{\dagger}\be_{0,n}\|^2
    \nonumber \\ &
    {}={} \Ex\|\bM^{\dagger}\be_{0,n}\|^2 \frac{|\mathcal{N}_n|-1}{|\mathcal{N}_n|}
    \,,
\end{align}
where "$\setminus$" denotes the set difference operator.

Since the noise statistics are spatially invariant, see \textsf{A5},
\begin{align}
	\Ex \|\bM^{\dagger}\be_{0,n}\|^2   &   {}={}
    \Ex \big\{ \tr\{\bM^{\dagger} \be_n\be_n^\top (\bM^{\dagger})^\top \} \big\}
    \nonumber \\ & 
    {}={} \|\bM^{\dagger} \bSigma_{\be}^{1/2}\|_F^2
\end{align}
and, by approximating the superpixel sizes by their average value (i.e. $|\mathcal{N}_n|\simeq S$), we can approximate~\eqref{eq:regDerA_noise_term2} as
\begin{align}
	& \Ex\bigg\{\frac{1}{N}\sum_{n=1}^N\|\bM^{\dagger}\big(\be_n-\frac{1}{|\mathcal{N}_n|}\sum_{i\in\mathcal{N}_n}\be_i)\|^2 \bigg\}
    \nonumber \\ & \hspace{1.5cm}
    {}\simeq{} \|\bM^{\dagger} \bSigma_{\be}^{1/2}\|_F^2 \frac{S-1}{S} \,.
\end{align}



Finally, approximating the expectations with respect to $\ba_n$ and $\psi_n$ by their instantaneous values and using the estimates $\hat{\ba}_{\mathcal{D}_n}$ and $\hat{\psi}_{\mathcal{C}_n}$ obtained as solutions to the optimization problem~\eqref{eq:khype_opt_coarse_constr}, equation \eqref{eq:expected_constr_aad_2} can be approximated as
\begin{align} \label{eq:approx_2nd_eq_constraint}
	& \frac{1}{N}\sum_{n=1}^N \|\ba_n -\widehat{\ba}_{\mathcal{D}_n}\|^2 
    {}\simeq{} 
    \frac{1}{N}\sum_{n=1}^N \|\bM^{\dagger}(\by_n-\by_{\mathcal{D}_n})\| ^2
    \nonumber \\
    & - \frac{1}{N}\sum_{n=1}^N\|\bM^{\dagger}(\psi_n(\bM)  - \widehat{\psi}_{\mathcal{C}_n}(\bM)\big)\|^2
    \nonumber \\ & 
    - \|\bM^{\dagger} \bSigma_{\be}^{1/2}\|_F^2 \frac{S-1}{S} \,.
\end{align}

\subsection{The updated fine scale optimization problem}

Using these results, we can substitute the second quadratic equality constraint of problem~\eqref{eq:khype_opt_finescale_constr} by \eqref{eq:approx_2nd_eq_constraint}, resulting in the following problem
\begin{align} \label{eq:khype_opt_finescale_constr_2}
    &  \mathop{\arg\min}_{\{\ba_{n},\psi_{n},\bxi_n\}} \,\,\, 
    \frac{1}{2} \sum_{n=1}^N \|\psi_{n}\|_{\mathcal{H}}^2
    \\
    & \text{subject to} 
    \hspace{1ex} \ba_{n} \geq\cb{0} \,, \,\, \cb{1}^\top\ba_{n} = 1 \,, \,\,\, n=1,\ldots,N \,,
    \nonumber \\ 
    & \hspace{10ex} \bxi_{n} {}={} \by_{n} - \bM \ba_{n} - \psi_{n}(\bM) \,,
    \,\, n=1,\ldots,N \,,
    \nonumber \\
    & \hspace{10ex} \frac{1}{N}\sum_{n=1}^N\|\bxi_n\|_2^2 {}={} C_1 \,,
    \nonumber \\
    & \hspace{10ex}  \frac{1}{N}\sum_{n=1}^N
    \|\ba_n-\widehat{\ba}_{\mathcal{D}_n}\|_2^2 {}={} C_Y - C_E
    \nonumber \\ \nonumber
    & \hspace{10ex} - \frac{1}{N}\sum_{n=1}^N\|\bM^{\dagger}(\psi_n(\bM)  - \widehat{\psi}_{\mathcal{C}_n}(\bM)\big)\|^2 \,,
\end{align}
where constants $C_Y$ and $C_E$ are defined as
\begin{align} \label{eq:constants_CY_CE}
\begin{split}
	C_Y & {}={} \frac{1}{N}\sum_{n=1}^N \|\bM^{\dagger}(\by_n-\by_{\mathcal{D}_n})\|^2 \,,
    \\
    C_E & {}={} \|\bM^{\dagger} \bSigma_{\be}^{1/2}\|_F^2 \frac{S-1}{S} \,.
\end{split}
\end{align}

Defining
\begin{align}
    \bxi_{\psi,n} = \bM^{\dagger}\big(\psi_n(\bM)  - \widehat{\psi}_{\mathcal{C}_n}(\bM)\big), \,\,\, n=1,\ldots,N
\end{align}
and multiplying the quadratic constraints by $N/2$, we can represent problem~\eqref{eq:khype_opt_finescale_constr_2} equivalently as
\begin{align} 
    \label{eq:khype_opt_finescale_constr_3}
    & \mathop{\arg\min}_{\{\ba_{n},\psi_{n},\bxi_n,\bxi_{\psi,n}\}} \,\, 
    \frac{1}{2} \sum_{n=1}^N \|\psi_{n}\|_{\mathcal{H}}^2
    \\
    & \text{subject to \,}
    \ba_{n} \geq\cb{0} \,, \,\, \cb{1}^\top\ba_{n} = 1 \,, \,\,\, n=1,\ldots,N \,,
    \nonumber \\ 
    & \hspace{10ex} \bxi_{n} {}={} \by_{n} - \bM \ba_{n} - \psi_{n}(\bM) \,, \,\,\, n=1,\ldots,N \,,
    \nonumber \\
    & \hspace{10ex} \frac{1}{2}\sum_{n=1}^N\|\bxi_n\|_2^2  {}={} \frac{N}{2} C_1 \,,
    \nonumber \\
    & \hspace{10ex}  \frac{1}{2}\sum_{n=1}^N
    \!\big(\|\ba_n-\widehat{\ba}_{\mathcal{D}_n}\|_2^2 + \|\bxi_{\psi,n}\|^2 \big)
	{}={} \! \frac{N}{2} (C_Y - C_E) \,,
    \nonumber \\ \nonumber 
    & \hspace{10ex}  \bxi_{\psi,n} = \bM^{\dagger}\big(\psi_n(\bM)  - \widehat{\psi}_{\mathcal{C}_n}(\bM)\big) \,, \,\, n=1,\ldots,N \,.
\end{align}
Problem~\eqref{eq:khype_opt_finescale_constr_3} can now be used instead of~\eqref{eq:khype_opt_finescale_constr} to perform unmixing in the original spatial scale.

\section{Solving the optimization problems} \label{sec:prob_opt_solution}

The quadratic equality constraints in~\eqref{eq:khype_opt_coarse_constr} and~\eqref{eq:khype_opt_finescale_constr_3} make the optimization problems non-convex. Furthermore, the functional form of the variables $\psi_{\mathcal{C}_i}$ and $\psi_n$ makes them hard to be optimized in their primal form.
Thus, we resort to a Lagrangian relaxation and solve the dual optimization problem, which is concave and finite-dimensional~\cite{boyd2004cvxbook,ScholkopfBook:2001}. Although the non-convexity of the constraints implies the possibility of a non-zero duality gap, we will show in Section~\ref{sec:opt_strong_dual} that strong duality holds under mild conditions for problems~\eqref{eq:khype_opt_coarse_constr} and~\eqref{eq:khype_opt_finescale_constr_3}. Thus, this approach incurs no loss of performance.

\subsection{The coarse scale dual problem}
\label{sec:coarse_lagrangian}

The Lagrangian of~\eqref{eq:khype_opt_coarse_constr} is given by
\begin{align} \label{eq:lagrangian_coarsescale_i}
    \mathcal{J}_{\calC} {}={} \! & \sum_{i=1}^K \bigg\{
    \frac{1}{2} \|\psi_{\calC_i}\|_{\mathcal{H}}^2
    + \|\ba_{\calC_i}\|_2^2
    + \frac{\mu_0}{2}\Big(\|\bxi_{\calC_i}\|_2^2 - C_0\Big)
    \nonumber \\ &
    + \lambda_{\calC_i}(\cb{1}^\top\ba_{\calC_i} - 1)
    - \bgamma_{\calC_i}^\top\ba_{\calC_i}
    \nonumber \\ &
    - \bbeta_{\calC_i}^\top \big(\bxi_{\calC_i}-\by_{\calC_i} + \bM\ba_{\calC_i} +\psi_{\calC_i}(\bM) \big)
    \bigg\} \,,
\end{align}
where $\mu_0$, $\bbeta_{\calC_i}$, $\lambda_{\calC_i}$ and $\bgamma_{\calC_i}\geq\cb{0}$ are the Lagrange multipliers.
The optimality conditions with respect to the primal variables are given by
\begin{align}
\begin{split}
	\ba_{\calC_i}^* & {}={} \bM^\top \bbeta_{\calC_i} + \bgamma_{\calC_i} - \lambda_{\calC_i} \cb{1} \,,
	\\
	\psi_{\calC_i}^* & {}={} \sum_{\ell=1}^L \beta_{{\calC_i},\ell} \kappa(\cdot,\widetilde{\bm}_{\ell}) \,,
	\\
	\bxi_{\calC_i}^* & {}={} \frac{1}{\mu_0} \bbeta_{\calC_i} \,.
\end{split}
\end{align}

By replacing these solutions into the Lagrangian in~\eqref{eq:lagrangian_coarsescale_i}, we can derive the dual optimization problem, which is given by
\begin{align} \label{eq:dual_problem0_i}
\begin{split}
	& \max_{\mu_0} \, \max_{\bomega_{\calC}}  \,\,  \sum_{i=1}^K  \Big(
    \bomega_{\calC_i}^\top \bB_{\calC}(\mu_0) \, \bomega_{\calC_i}
    + \bc_{\calC_i} \bomega_{\calC_i} - \frac{\mu_0}{2}C_0 \Big)
    \\ &
    \text{subject to } \bgamma_{\calC_i}\geq \cb{0} \,, \,\,\, i=1,\ldots,K \,,
\end{split}
\end{align}
where $\bomega_{\calC}=\big[\bomega_1^\top,\ldots,\bomega_K^\top\big]^\top$, $\bomega_{\calC_i}=\big[\bbeta_{\calC_i}^\top,\bgamma_{\calC_i}^\top,\lambda_{\calC_i}\big]^\top$ is the vector of dual variables, and $\bB_{\calC}$ and $\bc_{\calC_i}$ are given by
\begin{align}
	\bB_{\calC}(\mu_0) {}={} -\frac{1}{2}
    \left[\begin{array}{ccc}
    \bK+\frac{1}{\mu_0}\bI + \bM\bM^\top   &   \bM    &    -\bM\cb{1} \\
    \bM^\top   &   \bI   &   -\cb{1} \\
    -\cb{1}^\top\bM^\top   &   -\cb{1}^\top   &   P
    \end{array}\right] \,,
\end{align}
\begin{align}
	\bc_{\calC_i} {}={} 
    \left[\begin{array}{c|c|c}
    	\by_{\calC_i}^\top   &  \cb{0} & -1
    \end{array}\right]\,.
\end{align}

\subsection{The fine scale dual problem}
\label{sec:fine_lagrangian}

The Lagrangian of~\eqref{eq:khype_opt_finescale_constr_3} is given by
\begin{align} \label{eq:lagrangian_finescale_i}
    \mathcal{J}_{\calD} & {}={}  \!  \sum_{n=1}^N \bigg\{
    \frac{1}{2} \|\psi_{n}\|_{\mathcal{H}}^2 
    + \frac{\mu_1}{2}\Big(\|\bxi_n\|_2^2 - C_1\Big)
    + \lambda_n(\cb{1}^\top\ba_{n} - 1)
    \nonumber \\ &
    + \frac{\mu_2}{2} \Big(
    \|\ba_n-\widehat{\ba}_{\mathcal{D}_n}\|_2^2 - C_Y
    + \|\bxi_{\psi,n}\|^2 
    + C_E \Big)
    \nonumber \\ &
    + \bmu_{3,n}^\top \Big(\bM^{\dagger}(\psi_n(\bM)  - \widehat{\psi}_{\mathcal{C}_n}(\bM)\big) - \bxi_{\psi,n} \Big)
    \nonumber \\ &
    - \bbeta_{n}^\top \big(\bxi_{n}-\by_{n} + \bM\ba_n +\psi_n(\bM) \big)
    %
    - \bgamma_n^\top\ba_n
    \bigg\} \,,
\end{align}
where $\mu_1$, $\mu_2$, $\bmu_{3,n}$, $\bbeta_n$, $\bgamma_n$ and $\lambda_n$ are the Lagrange multipliers. 
Differentiating the optimality conditions with respect to the primal variables and equating the result to zero we obtain
\begin{align}
\begin{split}
	&\ba_n^* {}={} \widehat{\ba}_{\mathcal{D}_n} + \frac{1}{\mu_2} \Big( 
    \bM^\top \bbeta_n + \bgamma_n - \lambda_n \cb{1} \Big) \,,
	\\
	&\psi_n^*  {}={} \sum_{\ell=1}^L \beta_{n,\ell} \kappa(\cdot,\widetilde{\bm}_{\ell}) 
    - \sum_{\ell=1}^L \big[[\bmu_{3,n}]^\top\bM^{\dagger}\big]_{\ell} 
    \,\kappa(\cdot,\widetilde{\bm}_{\ell}) \,,
	\\
	&\bxi_{n}^*  {}={} \frac{1}{\mu_1} \bbeta_n \,,
    \\
	&\bxi_{\psi,n}^{*}  {}={} \frac{1}{\mu_2} \bmu_{3,n} \,,
\end{split}
\end{align}
where $[\,\cdot\,]_{\ell}$ denotes the $\ell$-th position of a vector.

Substituting the solution to the primal problem in the Lagrangian, we obtain the following dual problem
\begin{align} \label{eq:dual_problem2_i}
	& \max_{\mu_1,\mu_2} \, \max_{\bomega} \,\,\, 
    \sum_{n=1}^N \Big( \bomega_n^\top \bB(\mu_1,\mu_2) \bomega_n + \bc_n \bomega_n  \Big)
    \nonumber \\ & \hspace{11ex}
    - \frac{N}{2} \big(\mu_1 C_1 + \mu_2 C_Y - \mu_2 C_E \big)
    \\ &
    \text{subject to } \bgamma_n \geq \cb{0} \,, \,\,\, n=1,\ldots,N \,,
    \nonumber
\end{align}
where $\bomega = \big[\bomega_1^\top,\ldots,\bomega_N^\top\big]^\top$ is a vector containing the dual variables, with entries given by $\bomega_n = \big[\bbeta_n^\top,\,\bmu_{3,n}^\top,\,\bgamma_n^\top,\,\lambda_n^\top\big]^\top$. The terms $\bB(\mu_1,\mu_2)$ and $\bc_n$ are defined in \eqref{eq:dual_prob_mtx_i}.

\begin{figure*} [htb]
\begin{align} \label{eq:dual_prob_mtx_i}
\begin{split}
    & \bB(\mu_1,\mu_2)  {}={}
    - \frac{1}{2} 
    \left[\begin{array}{cccc}
    \bK + \frac{1}{\mu_1}\bI + \frac{1}{\mu_2} \bM\bM^\top  & -\bK\big(\bM^{\dagger}\big)^\top   &   \frac{1}{\mu_2}\bM   &   -\frac{1}{\mu_2}\bM\cb{1} \\
    -\bM^{\dagger}\bK   &   \frac{1}{\mu_2}\bI + \bM^{\dagger} \bK \big(\bM^{\dagger}\big)^\top   &  \cb{0}  &  \cb{0} \\
    \frac{1}{\mu_2} \bM^\top   &  \cb{0}  &  \frac{1}{\mu_2}\bI  &  -\frac{1}{\mu_2}\cb{1} \\
    - \frac{1}{\mu_2} \cb{1}^\top\bM^\top   &   \cb{0}  &  -\frac{1}{\mu_2}\cb{1}^\top  &  \frac{1}{\mu_2} P
    \end{array}\right] \,, 
    \quad 
    \bomega_n {}={}
    \left[\begin{array}{c}
    	\bbeta_n \\ \bmu_{3,n} \\ \bgamma_n \\ \lambda_n
    \end{array}\right] \,,
    \\[0.15cm]
    & \hspace{3.5cm} 
    \bc_n  {}={} 
    \left[\begin{array}{c|c|c|c}
    	-\ba_{\mathcal{D}_n}^\top \bM^\top + \by_n^\top
        & - \psi_{\mathcal{C}_n}(\bM)^\top \big(\bM^{\dagger}\big)^\top
        & - \ba_{\mathcal{D}_n}^\top
        &  \cb{1}^\top \ba_{\mathcal{D}_n} -1 
    \end{array}\right] \,.
\end{split}
\end{align}
\end{figure*}


Although~\eqref{eq:dual_problem0_i} and~\eqref{eq:dual_problem2_i} being Lagrangian dual problems implies that they are concave with respect to all variables~\cite{boyd2004cvxbook}, they are still nonlinear and thus computationally intensive to solve given their large dimension. However, the cost function of~\eqref{eq:dual_problem2_i} (resp. \eqref{eq:dual_problem0_i}) becomes quadratic when $\mu_1$ and $\mu_2$ (resp. $\mu_0$) are fixed. This will allow us to propose an efficient algorithm in Section~\ref{sec:bisection_dual_probs} to solve these problems.

\subsection{Strong duality of the optimization problems}
\label{sec:opt_strong_dual}

Since optimization problems~\eqref{eq:khype_opt_coarse_constr} and~\eqref{eq:khype_opt_finescale_constr_3} are non-convex, it does not immediately follows that the Lagrangian duality gap is zero. This means that, unless shown otherwise, the optimal solutions to the dual problems in~\eqref{eq:dual_problem2_i} and~\eqref{eq:dual_problem0_i} can be different from those of~\eqref{eq:khype_opt_coarse_constr} and~\eqref{eq:khype_opt_finescale_constr_3}. Fortunately, building upon results from non-convex optimization in~\cite{tuy2013strongDualityQPQC} we can show that strong duality holds for this problem. This is formalized in the following result:
\begin{theorem} \label{thm:thm1_duality_ours}
Suppose that the variables $\mu_1^*$ and $\mu_2^*$ (resp. $\mu_0^*$) that solve problem~\eqref{eq:dual_problem2_i} (resp.~\eqref{eq:dual_problem0_i}) are strictly positive. Then, strong duality holds for problem~\eqref{eq:khype_opt_finescale_constr} (resp.~\eqref{eq:khype_opt_coarse_constr}).
\end{theorem}
\begin{proof}
The proof builds upon the results in~\cite[Theorem 6]{tuy2013strongDualityQPQC}. Due to space limitations, it is relegated to the supplemental material, also available in~\cite{borsoi2019BMUAN_arxiv}. 
\end{proof}
This shows that the proposed Lagrangian relaxation strategy can achieve the same solution to the original problems, which was the case in all our experiments.

\subsection{An efficient solution to the Lagrangian dual problem}
\label{sec:bisection_dual_probs}

In order to devise an efficient algorithm for solving problems~\eqref{eq:dual_problem0_i} and~\eqref{eq:dual_problem2_i}, we first note that the purpose of the maximization with respect to $\mu_0$, $\mu_1$ and $\mu_2$ is to ensure that the quadratic equality constraints in the primal problems~\eqref{eq:khype_opt_coarse_constr} and~\eqref{eq:khype_opt_finescale_constr_3} are satisfied. Thus, we will attempt to write the dual problem in an equivalent form that will allow us to exploit this property to obtain a simpler solution. 

The optimality conditions for the coarse scale dual problem~\eqref{eq:dual_problem0_i} with respect to $\mu_0$ are obtained by differentiating the cost function and setting it equal to zero, which gives
\begin{align} \label{eq:optimality_mu0_dual_i}
\begin{split}
	g_0(\mu_0;\bomega) = 0 \,,
\end{split}
\end{align}
where function $g_0$ is defined as
\begin{align}
	g_0(\mu_0;\bomega) {}={}  & 
	\frac{1}{\mu_0^2}
	\sum_{i=1}^K \|\bbeta_{\calC_i}\|_2^2
	- K\,C_0
	\,.
\end{align}

Similarly, for the fine scale dual problem~\eqref{eq:dual_problem2_i} the optimality conditions with respect to $\mu_1$ and $\mu_2$ will be given by
\begin{align} \label{eq:optimality_mu_dual_i}
\Bigg\{
\begin{split}
	g_1(\mu_1,\mu_2;\bomega) = 0 \,,
    \\
    g_2(\mu_1,\mu_2;\bomega) = 0 \,,
\end{split}
\end{align}
where functions $g_1$ and $g_2$ are defined as
\begin{align}
	g_1(\mu_1,\mu_2;\bomega) {}={}  &  \frac{1}{\mu_1^2} \sum_{n=1}^N\|\bbeta_n\|_2^2 - N\, C_1 \,,
    \nonumber \\ 
    g_2(\mu_1,\mu_2;\bomega) {}={}  & \frac{1}{\mu_2^2} \sum_{n=1}^N \Big( 
    \|\bM^\top \bbeta_n + \bgamma_n - \lambda_n \cb{1}\|_2^2
    + \|\bmu_{3,n}\|_2^2
    \Big)
    \nonumber \\ &
	- N \big( C_Y - C_E\big) \,.
\end{align}

Let us define the following functions
\begin{align} \label{eq:optimal0_w_cond_mu}
	\widetilde{\bomega}_{\calC}(\mu_0) = \mathop{\arg\max}_{\bomega_{\calC}\,:\,\bgamma_{\calC_i}\geq\cb{0}}
    & \,\, \sum_{i=1}^K \Big( \bomega_{\calC_i}^\top \bB_{\calC}(\mu_0) \bomega_{\calC_i} + \bc_{\calC_i} \bomega_{\calC_i}  \Big)
\end{align}
for the coarse scale problem, and
\begin{align} \label{eq:optimal_w_cond_mu}
	\widetilde{\bomega}(\mu_1,\mu_2) = \mathop{\arg\max}_{\bomega\,:\,\bgamma_n\geq\cb{0}}
    & \sum_{n=1}^N \Big( \bomega_n^\top \bB(\mu_1,\mu_2) \bomega_n + \bc_n \bomega_n  \Big)
\end{align}
for the fine scale problem.
By substituting~\eqref{eq:optimal0_w_cond_mu} in $g_0$ and~\eqref{eq:optimal_w_cond_mu} in $g_1$ and $g_2$, the optimal $\mu_0$, $\mu_1$ and $\mu_2$ can be found by solving two systems of equations, one for the coarse scale problem, given by
\begin{align} \label{eq:nonlin_eq0_sys_bisec1}
	\text{find } \mu_0 \,\,\, \text{such that} \,\,\, 
    g_0(\mu_0;\widetilde{\bomega}_{\calC}) = 0 \,,
\end{align}
and another for the fine scale problem, given by
\begin{align} \label{eq:nonlin_eq_sys_bisec1}
	\text{find } \mu_1,\mu_2 \,\,\, \text{such that} \,\,\, 
    \left\{ \begin{array}{cc}
    g_1(\mu_1,\mu_2;\widetilde{\bomega}) = 0 \,, \\
    g_2(\mu_1,\mu_2;\widetilde{\bomega}) = 0 \,,
    \end{array}\right.
\end{align}
where we omitted the dependency of functions $\widetilde{\bomega}_{\calC}$ and $\widetilde{\bomega}$ on $\mu_0$, $\mu_1$ and $\mu_2$ for notational simplicity.
This consists in maximizing the inner optimization problems in~\eqref{eq:dual_problem0_i} and~\eqref{eq:dual_problem2_i} (with respect to $\bomega_{\calC}$ or $\bomega$) such that the quadratic constraints of the primal problems are satisfied.

Although many techniques can be used to solve \eqref{eq:nonlin_eq0_sys_bisec1} (resp. \eqref{eq:nonlin_eq_sys_bisec1}), one must note that evaluating functions $\widetilde{\bomega}_{\calC}(\mu_0)$ (resp. $\widetilde{\bomega}(\mu_1,\mu_2)$) is computationally expensive. Thus, to have an efficient solution we resort to a bisection strategy, which is a robust algorithm that converges to approximate solutions to these problems with relatively few function evaluations.

Although the solution of \eqref{eq:nonlin_eq0_sys_bisec1} using the conventional bisection algorithm is straightforward, the multidimensional case is less clear. Thus, we present it in the remaining of this section.
The multidimensional bisection algorithm relies on the Poincar\'e-Miranda theorem, which states that if a set of multivariate functions change sign in an interval for any of its coordinates, then there is at least one root, common to all such functions, within that interval~\cite{galvan2017multivariateBisection,bachrathy2012multidim_bisection}. 
This condition can be used to verify whether a given region in the function's domain contains a zero or not. 

Thus, by defining a search space and dividing it in two parts along one of the coordinates, we can test to see in which half the root is contained. By performing this operation alternately along each of the function coordinates, we can get arbitrarily close to the root~\cite{galvan2017multivariateBisection,bachrathy2012multidim_bisection}. This procedure is detailed in Algorithm~\ref{alg:2d_bisection}.
In all our experiments, we ran Algorithm~\ref{alg:2d_bisection} for ten iterations or until the relative variation of the parameters became smaller than a tolerance factor~$\epsilon=0.1$.
The final Blind Multiscale Unmixing Algorithm for Nonlinear spectral unmixing (BMUA-N) is presented in Algorithm~\ref{alg:proposed_alg_1}.

\begin{algorithm} [bth]
\small
\SetKwInOut{Input}{Input}
\SetKwInOut{Output}{Output}
\caption{Bi-dimensional bisection algorithm~\label{alg:2d_bisection}}
\Input{Functions $g_1,g_2:\mathbb{R}^2\to\mathbb{R}$.}
\Output{The estimated root $(a_r,b_r)$.}
Define an initial rectangle containing the root $R=\{(a_1,b_1),(a_2,b_1),(a_1,b_2),(a_2,b_2)\}$, $a_1<a_2$, $b_1<b_2$ \;
\While{Stopping criteria is not satisfied}{
Compute centers: $(a_c,b_c)=\big((a_2-a_1)/2,(b_2-b_1)/2\big)$ \;

Divide the search space in two and check for the root:
$R'=\{(a_1,b_1),(a_c,b_1),(a_1,b_2),(a_c,b_2)\}$ \;
Evaluate $g_1$ and $g_2$ at the four vertices of $R'$\;
\uIf{the sign of both $g_1$ and $g_2$ is not constant at all vertices of $R'$}
{$a_2=a_c$ \; }
\Else{$a_1=a_c$}

Partition rectangle across the other dimension
$R'=\{(a_1,b_1),(a_2,b_1),(a_1,b_c),(a_2,b_c)\}$ \;
Evaluate $g_1$ and $g_2$ at the four vertices of $R'$\;
\uIf{the sign of both $g_1$ and $g_2$ is not constant at all vertices of $R'$}
{$b_2=b_c$ \; }
\Else{$b_1=b_c$}
}
$(a_r,b_r)=\big((a_2-a_1)/2,(a_2-a_1)/2\big)$ \;
\KwRet $(a_r,b_r)$\;
\end{algorithm}

\begin{algorithm} [bth]
\small
\SetKwInOut{Input}{Input}
\SetKwInOut{Output}{Output}
\caption{BMUA-N~\label{alg:proposed_alg_1}}
\Input{$\bY$, $\bM$, ${\sigma_{\be,{\psi}}^2}$, the number of superpixels $K$ and multiscale decomposition matrix $\bW$.}
\Output{The estimated abundance matrix $\widehat{\!\bA}$.}
Estimate the noise covariance matrix $\bSigma_{\be}$ from $\bY$ \;
Compute the constants $C_0$, $C_1$, $C_Y$ and $C_E$ using equations~\eqref{eq:constant_C0},~\eqref{eq:constant_C1}, and~\eqref{eq:constants_CY_CE} \;
Compute $\bY_{\!\mathcal{C}}=\bY\bW$\; 
Find $\widehat{\!\bA}_{\mathcal{C}}$ by solving \eqref{eq:khype_opt_coarse_constr} using the procedures detailed in Sections~\eqref{sec:coarse_lagrangian} and~\eqref{sec:bisection_dual_probs}\;
Compute $\widehat{\!\bA}_{\mathcal{D}}$ using~\eqref{eq:mscale_conjugate_transf}\;
Find $\widehat{\!\bA}$ by solving \eqref{eq:khype_opt_finescale_constr_3} using the procedures detailed in sections~\eqref{sec:fine_lagrangian} and~\eqref{sec:bisection_dual_probs}\;
\KwRet $\widehat{\!\bA}$\;
\end{algorithm}

\section{Determining the number of superpixels}
\label{sec:superpixel_size_fnder}


A parameter of fundamental importance in the design of the proposed multiscale transform $\bW$ is the number of superpixels $K$, or, equivalently, the average size of each superpixel $S=N/K$.
The purpose of the multiscale transform is to group semantically/spectrally similar pixels, which are then averaged (within each superpixel) to constitute the coarse scale image, capturing spatial correlation and reducing the influence of noise.
From this definition, the desired average superpixel size could be intuitively defined as the largest value of $S$ such that the superpixels are still spectrally homogeneous.

In order to evaluate the homogeneity of the superpixels, we consider the distribution of the singular values of the sets of pixels within each superpixel, which are ordered in the form of matrices (matricized). Thus, the $j$-th matricized superpixel $\bY_{\!j}$ can be written as
\begin{align}
	\bY_{\!j} {}={} \big[\,\by_{I_1},\ldots,\by_{I_{|\mathcal{N}_j|}} \big] \,, 
    \quad
    \{I_1,\ldots,I_{|\mathcal{N}_j|} \} \subseteq \mathcal{N}_j \,,
\end{align}
for $j=1,\ldots,K$. Denote the singular values of $\bY_{\!j}$ by $\rho_{j,1},\rho_{j,2},\ldots,\rho_{j,|\mathcal{N}_j|}$, ordered from the largest to the smallest magnitude.
%
%
The homogeneity of the $j$-th superpixel can then be assessed using the ratio between the two largest singular values of $\bY_{\!j}$, which intuitively evaluates how close $\bY_{\!j}$ is to being a rank-1 matrix. This measure has already been successfully employed to detect heterogeneous superpixels in HS segmentation~\cite{yi2018superpixelHomogeneityTesting}. We then define the average homogeneity of all superpixels $\Homogeneity(K)$ as a function of the number $K$ of superpixels as
\begin{align}
	\Homogeneity(K) = \frac{1}{K} \sum_{j=1}^K \frac{|\rho_{j,1}|}{|\rho_{j,2}|} \,,
\end{align}
where we assume that $\rho_{j,2}$ exists and is nonzero.
%
Thus, $K$ can be selected using the following simple criterion:
\begin{align}
\begin{split}
    K {}={} & \max \,\,\, j
    \\ &
    \text{subject to } \, \Homogeneity(j) \geq (1-\varepsilon) \max_v \big\{\Homogeneity(v)\big\} \,,
    \\ &
    \hspace{10.5ex} K_{\min} \leq j \leq K_{\max} \,,
\end{split}
\end{align}
where we restrict the number of superpixels to be within a prescribed interval $[K_{\min}, K_{\max}]$.

\section{Computational complexity analysis} \label{sec:complexity}

The computational complexity of the proposed algorithm depends mainly on the two bisection procedures in Algorithm~\ref{alg:proposed_alg_1}. Each iteration of the bisection method in the coarse domain problem involves solving $K$ quadratic problems (QPs) in $L+P+1$ variables (one for each superpixel), whereas each iteration of the bisection method in the original image scale involves solving $N$ QPs in $L+2P+1$ variables (one for each pixel).
Since the bisection method reduces the search domain by half at each iteration, it converges linearly~\cite{bachrathy2012multidim_bisection}. Nevertheless, this is sufficient to achieve a reasonable approximation of the optimal parameters $\mu_0^*$, $\mu_1^*$ and $\mu_2^*$ in relatively few iterations~($\leq10$ in our experiments).
%
To see how this compares to other spatially regularized methods, consider for instance the TV-based nonlinear SU algorithm in~\cite{chen2014nonlinear2}. This algorithm employs a variable splitting procedure that leads to an iterative algorithm. At each iteration, $N$ QPs in $L+P+1$ variables (one for each pixel) and $P$ linear systems in $N$ variables (one for each endmember) must be solved, which is comparable to our method.
This illustrates how the proposed separation of the SU problem in two spatial scales maintains a computational complexity that is competitive with other algorithms, even when the estimation of the parameters is considered.

\section{Results} \label{sec:results}

In this section, we evaluate the performance of the proposed method using both synthetic and real datasets. The BMUA-N is compared with the fully constrained least squares (FCLS), with the unregularized K-Hype~\cite{chen2013nonlinear}, with the TV-based K-Hype (K-Hype-TV)~\cite{chen2014nonlinear2}, with the CDA-NL~\cite{halimi2016unmixingVariabilityNonlinearityMismodeling} and with the NDU~\cite{ammanouil2016nonlinear} algorithms.
The performances were evaluated using the Root Mean Squared Error (RMSE) between the estimated abundance maps ($\text{RMSE}_{\bA}$) and between the reconstructed images ($\text{RMSE}_{\bY}$). The RMSE between a true, generic matrix~$\bX$ and its estimate~$\widehat{\!\bX}$ is defined as
\begin{equation}
    \text{RMSE}_{\bX} = \sqrt{\textstyle{\frac{1}{N_{\!\bX}}}\|\bX- \widehat{\!\bX}\|^2_F} \,,
\end{equation}
where $N_{\!\bX}$ denotes the number of elements in the matrix~$\bX$.

For the proposed method, the noise covariance matrix $\bSigma_{\be}$ was estimated using the residual method described in~\cite{roger1996residualMethodNoiseCovarEstimation,mahmood2017modifiedNoiseCovarianceEstimation}, and the superpixel sizes were selected using the strategy detailed in Section~\ref{sec:superpixel_size_fnder}, with $K_{\min}=N/8$, $K_{\max}=N/170$ and $\varepsilon=0.1$. The polynomial kernel described in~\eqref{eq:polyKernel} was used with $d=2$ for all kernel-based non-linear SU algorithms.

\begin{table} [!ht]
\footnotesize
\centering
\caption{Quantitative results for data cubes DC1 and DC2.}
\vspace{-2ex}
\renewcommand{\arraystretch}{1.2}
\begin{tabular}{c||c|c|c|c|c}
\hline\hline
\multicolumn{6}{c}{DC1 data cube}  \\
\hline
& & \multicolumn{2}{c|}{BLMM} & \multicolumn{2}{|c}{PNMM}  \\
\hline
SNR	&	Method	&	$\text{RMSE}_{\!\bA}$	&	$\text{RMSE}_{\bY}$			&	$\text{RMSE}_{\!\bA}$	&	$\text{RMSE}_{\bY}$		\\
\hline				
\multirow{6}{*}{20dB} 
&	FCLS	&	0.2587	&	0.1143	&	0.1657	&	0.1038	\\
&	K-Hype	&	0.0575	&	0.0816	&	0.0972	&	0.0769	\\
&	K-Hype-TV	&	0.0371	&	0.0814	&	0.0800	&	\textbf{0.0766}	\\
&	CDA-NL	&	0.0730	&	0.0821	&	0.1708	&	0.0789	\\
&	NDU	&	0.0483	&	\textbf{0.0736}	&	0.1263	&	0.0802	\\
&	BMUA-N	&	\textbf{0.0326}	&	0.0820	&	\textbf{0.0730}	&	0.0774	\\

\hline													
\multirow{6}{*}{30dB} 
&	FCLS	&	0.2591	&	0.0836	&	0.1633	&	0.0731	\\
&	K-Hype	&	0.0346	&	0.0258	&	0.0765	&	\textbf{0.0242}	\\
&	K-Hype-TV	&	\textbf{0.0323}	&	0.0258	&	0.0757	&	0.0243	\\
&	CDA-NL	&	0.0485	&	0.0265	&	0.1625	&	0.0290	\\
&	NDU	&	0.0336	&	\textbf{0.0256}	&	0.1208	&	0.0323	\\
&	BMUA-N	&	0.0325	&	0.0257	&	\textbf{0.0734}	&	0.0243	\\

\hline\hline
\multicolumn{6}{c}{DC2 data cube} \\
\hline
& & \multicolumn{2}{c|}{BLMM} & \multicolumn{2}{|c}{PNMM} \\
\hline
SNR	&	Method	&	$\text{RMSE}_{\!\bA}$	&	$\text{RMSE}_{\bY}$			&	$\text{RMSE}_{\!\bA}$	&	$\text{RMSE}_{\bY}$	\\
\hline			
\multirow{6}{*}{20dB} 
&	FCLS	&	0.1718	&	0.1055	&	0.1554	&	0.1033	\\
&	K-Hype	&	0.0723	&	0.0774	&	0.1165	&	\textbf{0.0758}	\\
&	K-Hype-TV	&	0.0557	&	0.0771	&	0.1037	&	\textbf{0.0758}	\\
&	CDA-NL	&	0.0648	&	0.0780	&	0.1871	&	0.0781	\\
&	NDU	&	0.0518	&	\textbf{0.0710}	&	0.1363	&	0.0794	\\
&	BMUA-N	&	\textbf{0.0490}	&	0.0771	&	\textbf{0.1009}	&	0.0759	\\

\hline													
\multirow{6}{*}{30dB} 
&	FCLS	&	0.1714	&	0.0750	&	0.1553	&	0.0731	\\
&	K-Hype	&	0.0503	&	0.0244	&	0.0994	&	\textbf{0.0240}	\\
&	K-Hype-TV	&	0.0501	&	0.0244	&	0.0992	&	\textbf{0.0240}	\\
&	CDA-NL	&	0.0491	&	0.0250	&	0.1803	&	0.0289	\\
&	NDU	&	0.0422	&	\textbf{0.0243}	&	0.1318	&	0.0321	\\
&	BMUA-N	&	\textbf{0.0393}	&	0.0245	&	\textbf{0.0902}	&	0.0241	\\

\hline\hline
\end{tabular}
\label{tab:quantitative_results}
\end{table}


\begin{figure}[h]
\centering
\includegraphics[width=0.95\linewidth]{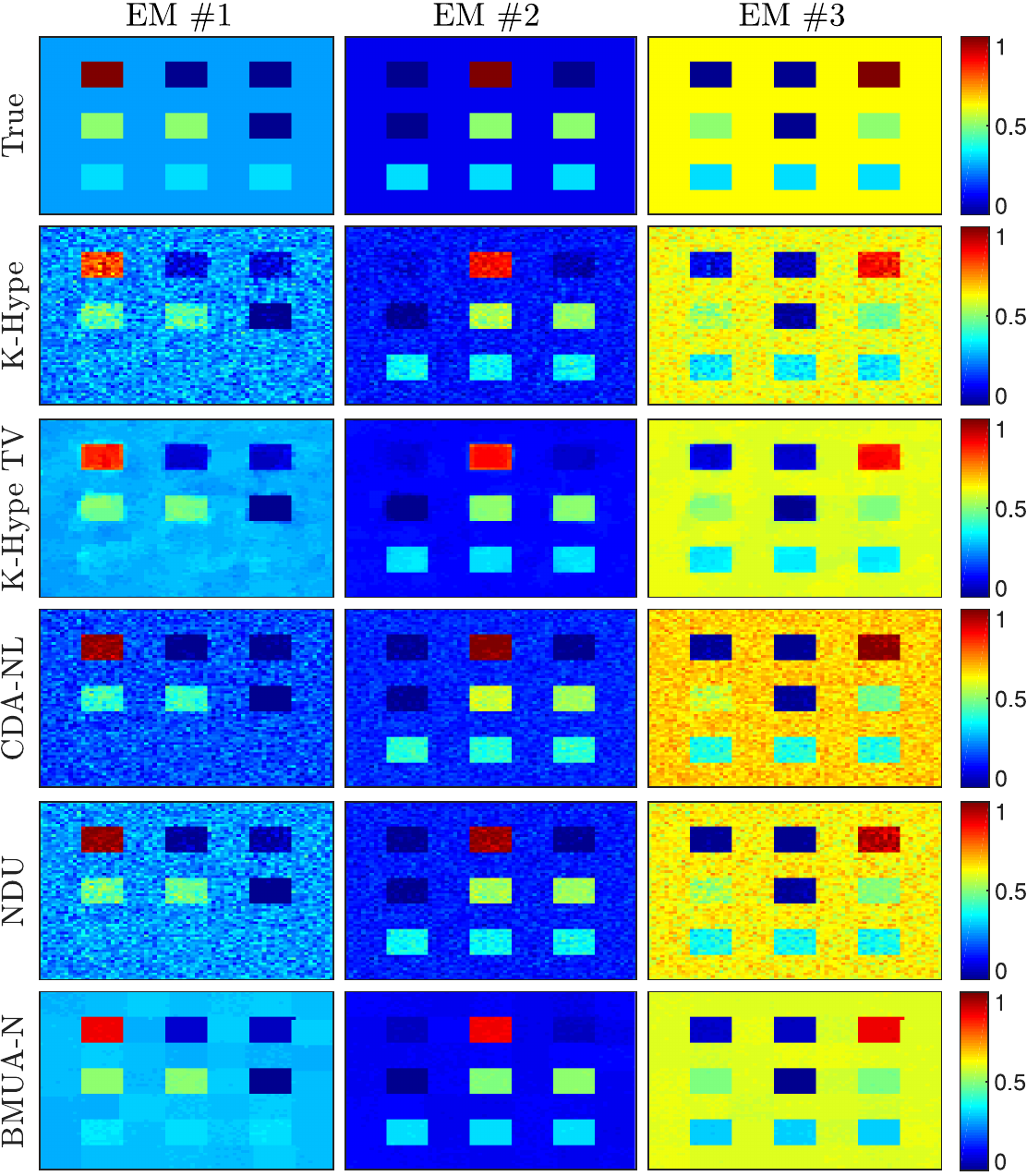}
\caption{Abundance maps estimated by all algorithms for the data \mbox{cube DC1.}}
\label{fig:abundances_synth_ex1}
\end{figure}
\begin{figure}[h]
\centering
\includegraphics[width=0.95\linewidth]{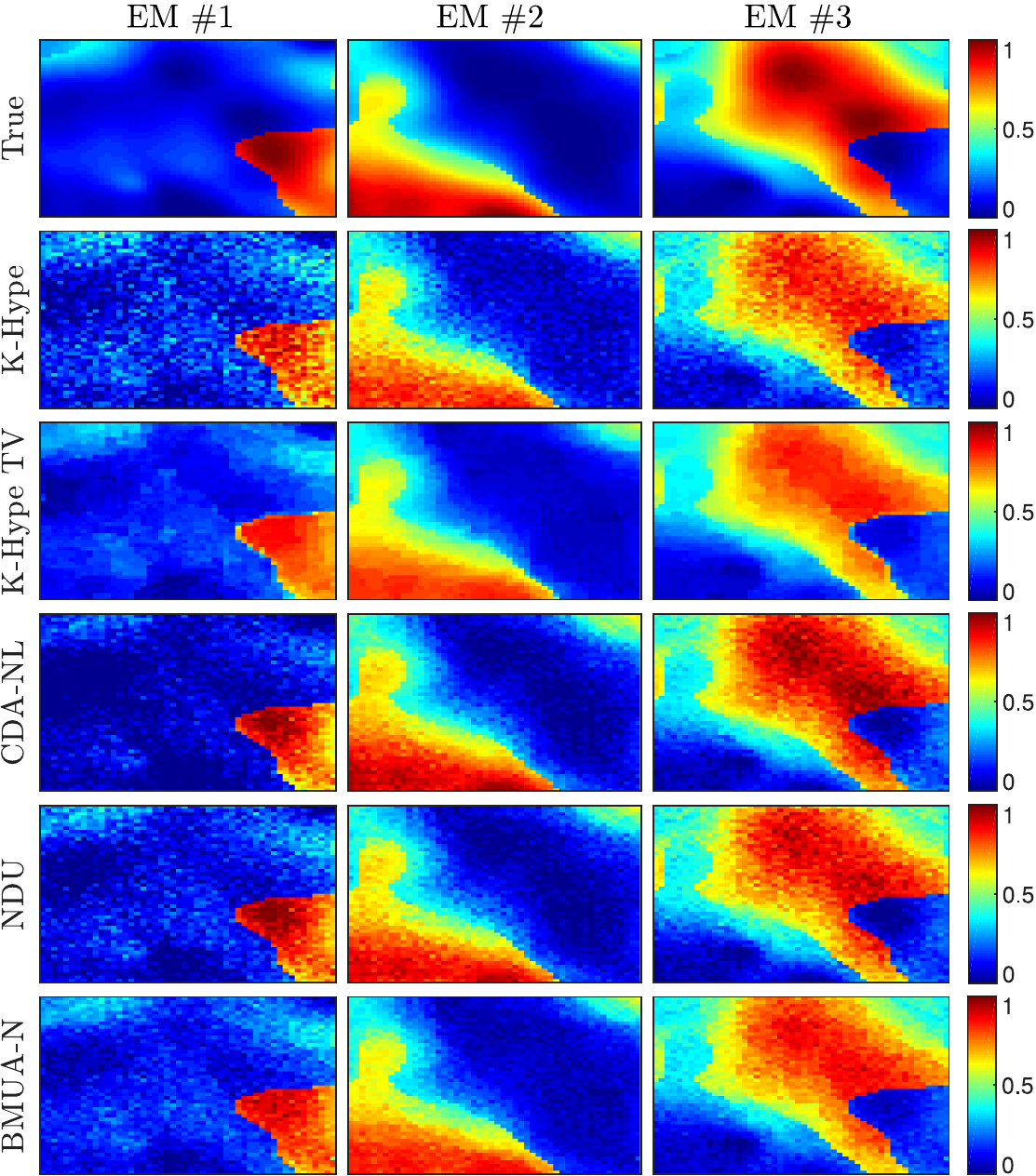}
\caption{Abundance maps estimated by all algorithms for the data \mbox{cube DC2.}}
\label{fig:abundances_synth_ex2}
\end{figure}

\subsection{Synthetic data sets}

To compare the performance of the different algorithms quantitatively, we created two synthetic datasets with spatially correlated abundance maps, namely Data Cube 0 (DC0), with $70\times70$ pixels, and Data Cube 1 (DC1), with $50\times50$ pixels. Both datasets were constructed using three spectral signatures with 224 bands extracted from the USGS Spectral Library.  The synthetic abundance maps are displayed in the first row of Figs.~\ref{fig:abundances_synth_ex1} and~\ref{fig:abundances_synth_ex2}.
The reflectance values were generated using two nonlinear mixture models, namely, the bilinear mixing model (BLMM), defined as
\begin{align}
	\by_n {}={} \bM\ba_n + \sum_{i=1}^{P-1} \sum_{j=i+1}^P  a_{n,i} a_{n,j} \bm_i \circ \bm_j + \be_n \,,
\end{align}
with $\circ$ being the Hadamard product, and the post-nonlinear mixing model (PNMM), defined as
\begin{align}
	\by_n {}={} (\bM\ba_n)^{0.7} + \be_n \,,
\end{align}
where a fixed exponential value of~$0.7$ has been applied to the LMM result.
Finally, white Gaussian noise with signal to noise ratios (SNR) of both 20 and 30dB was added to all datasets.

The parameters for each algorithm were either fixed or selected based on a grid search performed for each dataset, with search ranges defined based on those ranges discussed by the authors in the original publications. For the BMUA-N, we fixed the modeling errors as a small value relative to the average pixel energy, given by $\sigma_{\be,\psi}^2=10^{-8}\frac{1}{N}\sum_{n=1}^N\|\by_n\|_2^2$. For the K-Hype algorithm, we selected the parameter among the values $\mu\in\{0.001,0.002,0.005,0.01,0.02,0.1,1\}$. For the K-Hype-TV, the parameters were selected among the following values: $\mu\in\{0.001,0.002,0.005,0.01,0.02,0.1,1\}$, $\eta\in\{0,0.01,0.1,0.25,0.5,0.75,1\}$. For the NDU algorithm, the parameters were selected among the values $\lambda,\,\mu\in\{0.0005,0.005,0.05,0.5,5,50,500,5000\}$.

\subsubsection{Discussion}

The quantitative results of all algorithms are shown in Table~\ref{tab:quantitative_results}.
The proposed BMUA-N method outperformed the competing algorithms for almost all cases, except for one where its result was very close to that of the TV-based solution. This is despite the fact that the parameters of K-Hype, K-Hype-TV and NDU were selected through a grid search procedure.
The abundance maps provided by the nonlinear SU methods for both datacubes are displayed in Figures~\ref{fig:abundances_synth_ex1} and~\ref{fig:abundances_synth_ex2} for illustrative purposes for the case of the BLMM with an SNR of 20dB. The FCLS results were not displayed for the sake of space since they were significantly worse than those of the other algorithms. It can be seen that the BMUA-N results better approximates the ground truth, and even though the K-Hype-TV solution is smoother for DC2 its mean results are farther from the true values. Moreover, although the NDU achieves relatively good abundance reconstructions for the BLMM, the results are noisy and not as good for the PNMM.
The reconstruction errors of the nonlinear SU algorithms, also shown in Table~\ref{tab:quantitative_results}, were similar and significantly lower than those of the FCLS.
The execution times of the BMUA-N algorithm, shown in Table~\ref{tab:alg_exec_time}, are about 3.5 times higher than those of K-Hype-TV and half those of the NDU. Thus, the complexity of the BMUA-N is still on the same order of the other state of the art algorithms even though no significant parameter tuning is necessary.

\begin{table}[!t]
\footnotesize
\caption{Average execution time (in seconds) of the algorithms.}
\vspace{-0.2cm}
\centering
\renewcommand{\arraystretch}{1.2}
\begin{tabular}{c||c|c|c|c|c}
\hline\hline
&	DC1	&	DC2	& Cuprite & Urban & Jasper Ridge \\
\hline\hline
FCLS	&	0.59	&	0.30	&	11.66	&	0.34	&	1.37	\\
K-Hype	&	4.54	&	2.40	&	43.29	&	1.42	&	75.17	\\
K-Hype-TV	&	43.27	&	22.51	&	405.81	&	17.72	&	74.08	\\
CDA-NL	&	29.26	&	15.57	&	3295.99	&	27.90	&	137.63	\\
NDU	&	263.30	&	141.04	&	2260.45	&	74.32	&	483.34	\\
BMUA-N	&	144.60	&	79.03	&	1380.21	&	45.21	&	240.12	\\
\hline
\end{tabular}
\label{tab:alg_exec_time}
\end{table}

\begin{figure}
    \centering
    \includegraphics[height=0.29\linewidth]{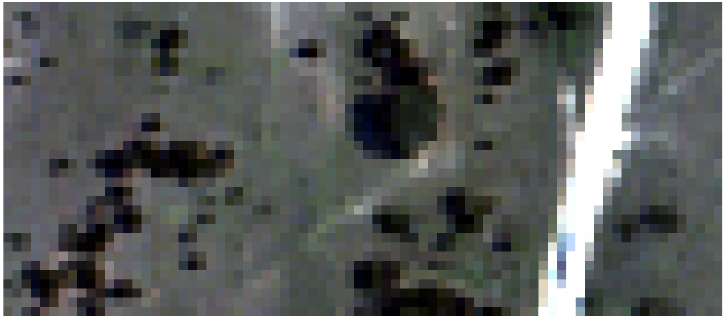}
    \includegraphics[height=0.29\linewidth]{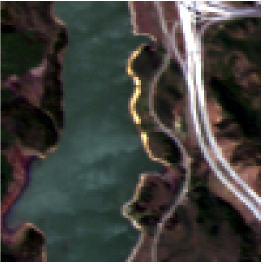}
    \caption{Subscene of the Urban (left) and Jasper Ridge (right) HIs.}
    \label{fig:exRealIm_HIs}
\end{figure}
\begin{figure}
\centering
\includegraphics[width=1\linewidth]{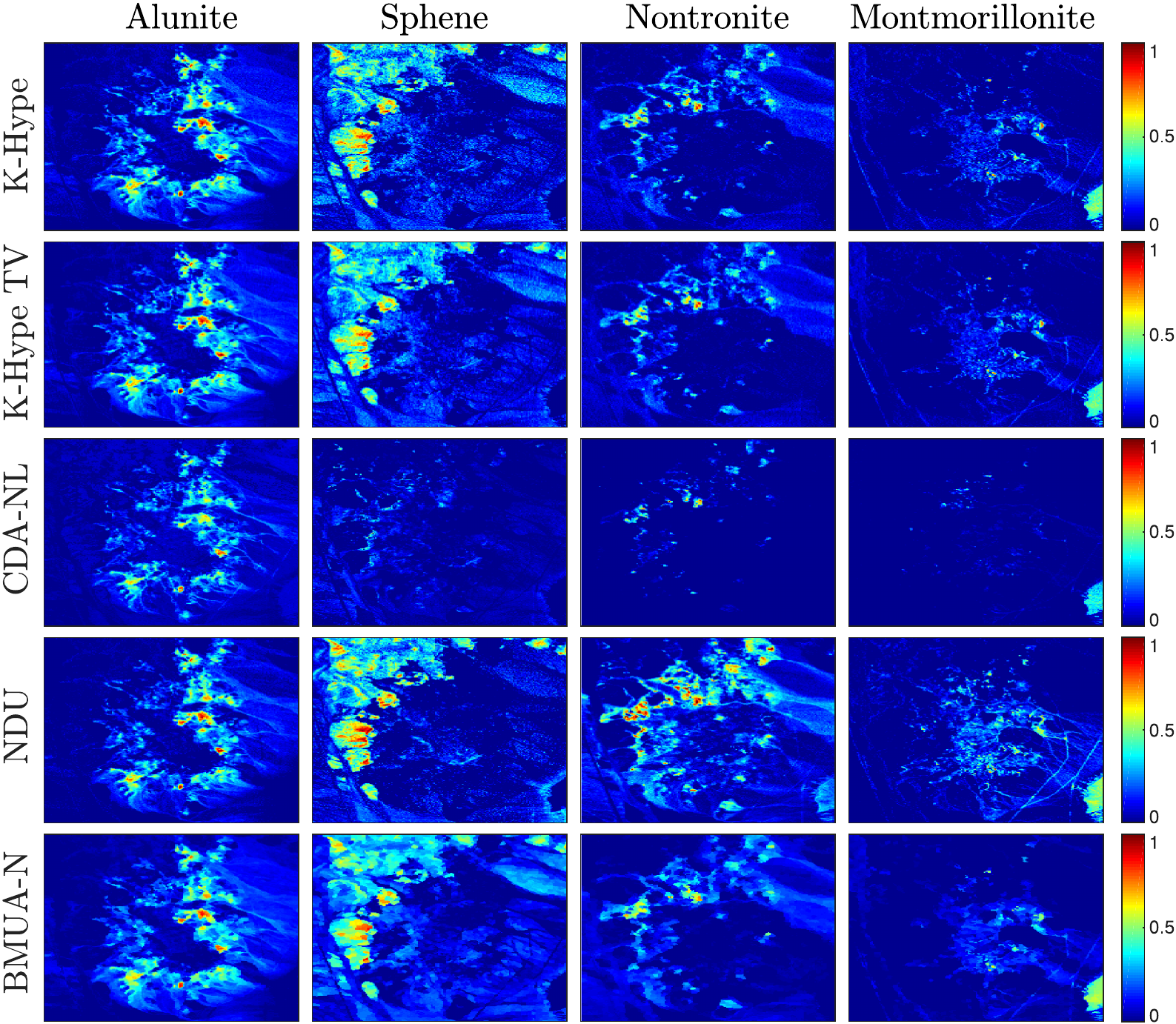}
\caption{Estimated abundance maps for four endmembers of the Cuprite image.}
\label{fig:exxReal}
\end{figure}
\begin{figure}
    \centering
    \includegraphics[width=1\linewidth]{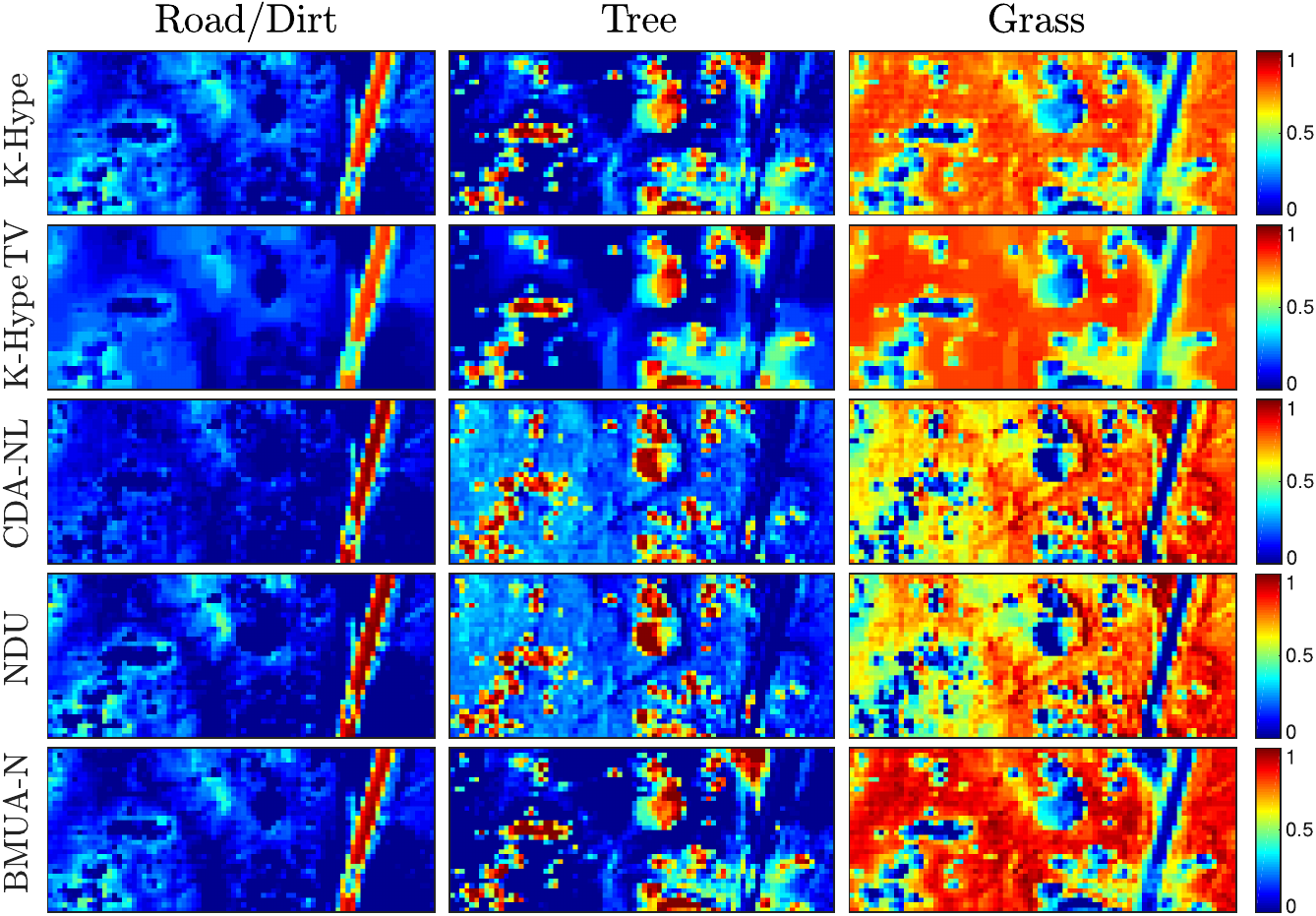}
    \caption{Estimated abundance maps for the Urban HI subscene.}
    \label{fig:exRealIm_urbanHI_abundances}
\end{figure}
\begin{figure}
    \centering
    \includegraphics[width=1\linewidth]{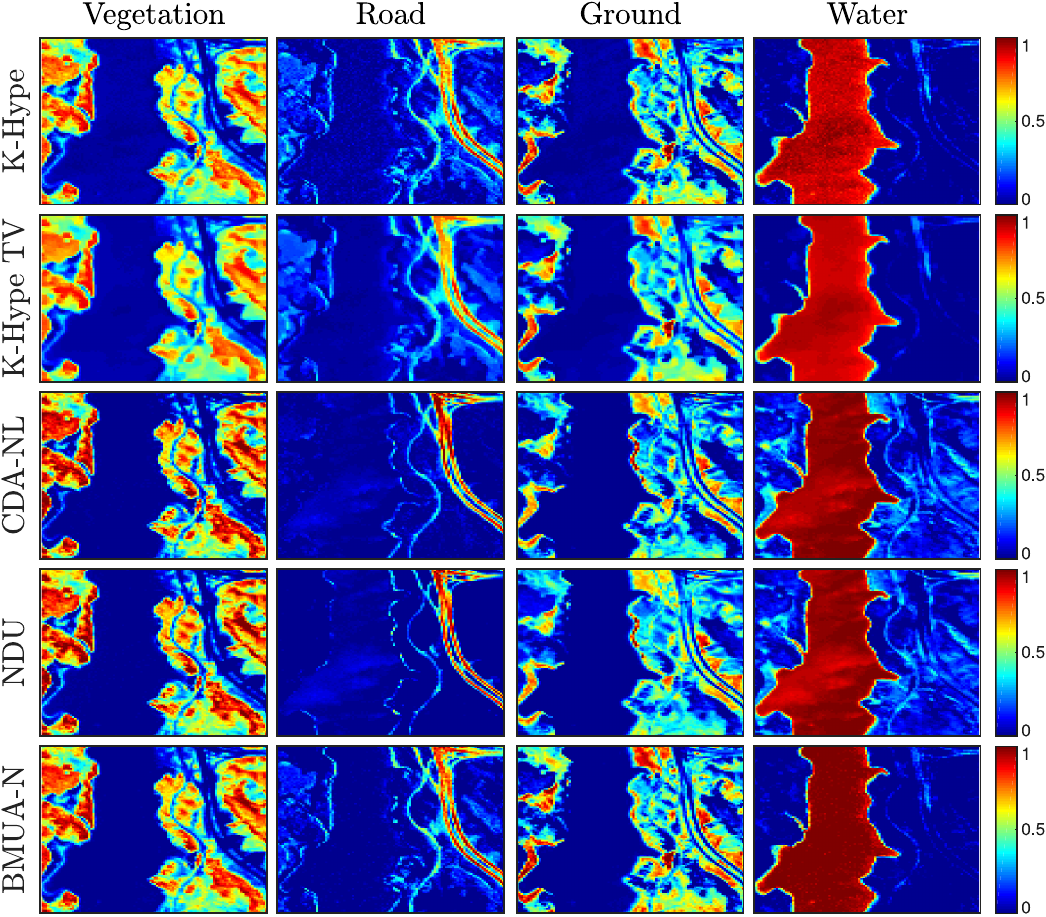}
    \caption{Estimated abundance maps for the Jasper Ridge HI.}
    \label{fig:exRealIm_jasperHI_abundances}
\end{figure}



%
\begin{table}[!t]
\footnotesize
\caption{Reconstruction errors ($\text{RMSE}_{\bY}$) for the real datasets.}
\vspace{-0.2cm}
\centering
\renewcommand{\arraystretch}{1.2}
\resizebox{\linewidth}{!}{%
\begin{tabular}{@{\hspace{0.4\tabcolsep}} c@{\hspace{0.5\tabcolsep}}
||@{\hspace{0.5\tabcolsep}}c @{\hspace{0.5\tabcolsep}}
|@{\hspace{0.5\tabcolsep}}c@{\hspace{0.5\tabcolsep}}
|@{\hspace{0.5\tabcolsep}}c@{\hspace{0.5\tabcolsep}}
|@{\hspace{0.5\tabcolsep}}c@{\hspace{0.5\tabcolsep}}
|@{\hspace{0.5\tabcolsep}}c@{\hspace{0.5\tabcolsep}}
|@{\hspace{0.5\tabcolsep}}c@{\hspace{0.5\tabcolsep}}}
\hline\hline
 & FCLS & K-Hype & K-Hype-TV & CDA-NL & NDU & BMUA-N \\
\hline\hline
Cuprite	&	0.0107	&	0.0082	&	0.0082	&	0.0095	&	\textbf{0.0080}	&	0.0090	\\
Urban	&	0.0228	&	0.0434	&	\textbf{0.0060}	&	0.0063	&	0.0159	&	0.0062	\\
Jasper Ridge	&	0.0225	&	0.0562	&	\textbf{0.0164}	&	0.0168	&	0.0167	&	0.0170	\\
\hline
\end{tabular}}
\label{tab:reconstr_err_real_img}
\end{table}

\subsection{Experiments with real data}

For the simulations with real data we consider the Cuprite, the Urban and the Jasper Ridge datasets, which were captured by the AVIRIS instrument and originally had 224 bands. Water absorption and low SNR bands were removed before processing, resulting in 188 bands for the Cuprite image, 162 bands for the Urban image and 198 bands for the Jasper Ridge image. Previous works indicate that 14 endmembers are present at the Cuprite mining field~\cite{Nascimento2005,imbiriba2018ULTRA_V,Borsoi_multiscaleVar_2018}, while the Jasper Ridge HI is known to have four predominant endmemers~\cite{imbiriba2018ULTRA_V}. For the Urban HI, we consider a smaller subscene (shown in Figure~\ref{fig:exRealIm_HIs}) containing three endmembers to allow for an easier evaluation.

The reconstructed abundance maps of the nonlinear SU algorithms for the three datasets are presented in Figures~\ref{fig:exxReal},~\ref{fig:exRealIm_urbanHI_abundances} and~\ref{fig:exRealIm_jasperHI_abundances}. For the Cuprite dataset, four endmembers were selected whose distribution could be clearly distinguished in the scene~\cite{Nascimento2005}.
The reconstructed abundance maps of the nonlinear SU algorithms are presented in Figure~\ref{fig:exxReal}, where it can be seen that, except for the case of the CDA-NL, the results for all algorithms are generally compatible and agree with previous studies of this scene~\cite{Nascimento2005,imbiriba2018ULTRA_V,Borsoi_multiscaleVar_2018}. Nevertheless, a careful analysis reveals that the BMUA-N results, displayed in the bottom row of the figure, show smoother abundance reconstructions without compromising image details and discontinuities. 
For the Urban HI, the last row of Figure~\ref{fig:exRealIm_urbanHI_abundances} shows that the abundances estimated by the BMUA-N generally contain stronger components for the road, tree and grass endmembers at the regions where these materials appear more prominently, leading to a better separation between the different endmember classes according to Figure~\ref{fig:exRealIm_HIs}.
The abundance maps estimated by all algorithms for the Jasper Ridge HI, shown in Figure~\ref{fig:exRealIm_jasperHI_abundances}, generally agree with the distribution of the corresponding materials observed in Figure~\ref{fig:exRealIm_HIs}. A careful analysis of the BMUA-N results shows that it has stronger vegetation and water components when compared to K-Hype and K-Hype-TV. Moreover, although the abundances estimated by CDA-NL and NDU contain slightly clearer results for the road endmember, they also show considerably more confusion between ground and water.

The reconstruction errors for all datasets are shown in Table~\ref{tab:reconstr_err_real_img}. Except for K-Hype, the nonlinear SU algorithms generally achieved much smaller reconstruction errors when compared to the FCLS. Although the results of BMUA-N were slightly higher than the TV-based solution, as can also be observed in the synthetic experiments, small variations in the reconstruction error do not necessarily correlate with better abundance reconstructions.
The execution times of the proposed method were again around 3.5 times higher than that of the TV-based solution but only about half of those of the NDU, which indicates that it scales favorably with \mbox{larger image sizes.}


\section{Conclusions} \label{sec:conclusions}

In this paper, a blind multiscale unmixing strategy was proposed for nonlinear kernel-based mixing models. Based on the concept of a multiscale regularization strategy recently introduced in~\cite{Borsoi_multiscale_lgrs_2018}, we were able to efficiently capture image spatial information by splitting the nonlinear mixing process between two image scales, one containing the coarse, low-dimensional image structures and another representing the original image domain.
Furthermore, we employed a theory-based statistical framework to devise a consistent strategy to automatically select the regularization parameters of the proposed algorithm and of the multiscale transformation. This resulted in a truly blind (from the parameter setting perspective) multiscale regularization framework.
The unmixing problem was formulated using quadratically constrained optimization problems, for which efficient solutions were obtained by exploring their strong duality and a reformulation of their dual representations as root-finding problems.
Simulation results with both synthetic and real data indicate that the proposed strategy leads to a consistent performance improvement when compared to the classical Total Variation regularization, even though no parameter adjustment is necessary.

\bibliographystyle{IEEEtran}
\bibliography{references,references_nonlSppx,ourpapers}

\begin{thebibliography}{10}
\providecommand{\url}[1]{#1}
\csname url@samestyle\endcsname
\providecommand{\newblock}{\relax}
\providecommand{\bibinfo}[2]{#2}
\providecommand{\BIBentrySTDinterwordspacing}{\spaceskip=0pt\relax}
\providecommand{\BIBentryALTinterwordstretchfactor}{4}
\providecommand{\BIBentryALTinterwordspacing}{\spaceskip=\fontdimen2\font plus
\BIBentryALTinterwordstretchfactor\fontdimen3\font minus
  \fontdimen4\font\relax}
\providecommand{\BIBforeignlanguage}[2]{{%
\expandafter\ifx\csname l@#1\endcsname\relax
\typeout{** WARNING: IEEEtran.bst: No hyphenation pattern has been}%
\typeout{** loaded for the language `#1'. Using the pattern for}%
\typeout{** the default language instead.}%
\else
\language=\csname l@#1\endcsname
\fi
#2}}
\providecommand{\BIBdecl}{\relax}
\BIBdecl

\bibitem{Bioucas-Dias-2013-ID307}
J.~M. Bioucas-Dias, A.~Plaza, G.~Camps-Valls, P.~Scheunders, N.~Nasrabadi, and
  J.~Chanussot, ``Hyperspectral remote sensing data analysis and future
  challenges,'' \emph{IEEE Geoscience and Remote Sensing Magazine}, vol.~1,
  no.~2, pp. 6--36, 2013.

\bibitem{Dobigeon-2014-ID322}
N.~Dobigeon, J.-Y. Tourneret, C.~Richard, J.~C.~M. Bermudez, S.~McLaughlin, and
  A.~O. Hero, ``Nonlinear unmixing of hyperspectral images: Models and
  algorithms,'' \emph{IEEE Signal Processing Magazine}, vol.~31, no.~1, pp.
  82--94, Jan 2014.

\bibitem{Imbiriba2016_tip}
T.~Imbiriba, J.~C.~M. Bermudez, C.~Richard, and J.-Y. Tourneret,
  ``Nonparametric detection of nonlinearly mixed pixels and endmember
  estimation in hyperspectral images,'' \emph{IEEE Transactions on Image
  Processing}, vol.~25, no.~3, pp. 1136--1151, March 2016.

\bibitem{somers2011variabilityReview}
B.~Somers, G.~P. Asner, L.~Tits, and P.~Coppin, ``Endmember variability in
  spectral mixture analysis: A review,'' \emph{Remote Sensing of Environment},
  vol. 115, no.~7, pp. 1603--1616, 2011.

\bibitem{imbiriba2018glmm}
T.~Imbiriba, R.~A. Borsoi, and J.~C.~M. Bermudez, ``Generalized linear mixing
  model accounting for endmember variability,'' in \emph{2018 IEEE
  International Conference on Acoustics, Speech and Signal Processing
  (ICASSP)}.\hskip 1em plus 0.5em minus 0.4em\relax Calgary, Canada: IEEE,
  2018, pp. 1862--1866.

\bibitem{Borsoi_2018_Fusion}
R.~A. {Borsoi}, T.~{Imbiriba}, and J.~C.~M. {Bermudez}, ``Super-resolution for
  hyperspectral and multispectral image fusion accounting for seasonal spectral
  variability,'' \emph{IEEE Transactions on Image Processing}, vol.~29, pp.
  116--127, 2020.

\bibitem{borsoi2019deepGun}
------, ``Deep generative endmember modeling: An application to unsupervised
  spectral unmixing,'' \emph{IEEE Transactions on Computational Imaging},
  vol.~6, pp. 374--384, 2020.

\bibitem{borsoi2019EMlibManInterpVAE}
R.~A. Borsoi, T.~Imbiriba, J.~C.~M. Bermudez, and C.~Richard, ``Deep generative
  models for library augmentation in multiple endmember spectral mixture
  analysis,'' \emph{arXiv preprint}, 2019.

\bibitem{Thouvenin_IEEE_TSP_2016_PLMM}
P.-A. Thouvenin, N.~Dobigeon, and J.-Y. Tourneret, ``Hyperspectral unmixing
  with spectral variability using a perturbed linear mixing model,'' \emph{IEEE
  Transactions on Signal Processing}, vol.~64, no.~2, pp. 525--538, Feb. 2016.

\bibitem{Ray1996}
T.~W. Ray and B.~C. Murray, ``Nonlinear spectral mixing in desert vegetation,''
  \emph{Remote Sensing of Environment}, vol.~55, no.~1, pp. 59--64, 1996.

\bibitem{heylen2014review}
R.~Heylen, M.~Parente, and P.~Gader, ``A review of nonlinear hyperspectral
  unmixing methods,'' \emph{IEEE Journal of Selected Topics in Applied Earth
  Observations and Remote Sensing}, vol.~7, no.~6, pp. 1844--1868, June 2014.

\bibitem{mustard1989photometric}
J.~F. Mustard and C.~M. Pieters, ``Photometric phase functions of common
  geologic minerals and applications to quantitative analysis of mineral
  mixture reflectance spectra,'' \emph{Journal of Geophysical Research: Solid
  Earth}, vol.~94, no. B10, pp. 13\,619--13\,634, 1989.

\bibitem{Guilfoyle2001}
K.~J. Guilfoyle, M.~L. Althouse, and C.-I. Chang, ``A quantitative and
  comparative analysis of linear and nonlinear spectral mixture models using
  radial basis function neural networks,'' \emph{{IEEE} Transactions on
  Geoscience and Remote Sensing}, vol.~39, no.~10, pp. 2314--2318, 2001.

\bibitem{Altmann_tip_2012}
Y.~Altmann, A.~Halimi, N.~Dobigeon, and J.~Y. Tourneret, ``Supervised nonlinear
  spectral unmixing using a postnonlinear mixing model for hyperspectral
  imagery,'' \emph{IEEE Transactions on Image Processing}, vol.~21, no.~6, pp.
  3017--3025, June 2012.

\bibitem{Li_svm_2005}
P.-X. Li, B.~Wu, and L.~Zhang, ``Abundance estimation from hyperspectral image
  based on probabilistic outputs of multi-class support vector machines,'' in
  \emph{2005 IEEE International Geoscience and Remote Sensing Symposium
  (IGARSS)}, vol.~6, July 2005, pp. 4315--4318.

\bibitem{Heylen:2011kc}
R.~Heylen, D.~Burazerovic, and P.~Scheunders, ``{Non-linear spectral unmixing
  by geodesic simplex volume maximization},'' \emph{IEEE Journal of Selected
  Topics in Signal Processing}, vol.~5, no.~3, pp. 534--542, 2011.

\bibitem{Heylen_2014}
R.~Heylen and P.~Scheunders, ``A distance geometric framework for nonlinear
  hyperspectral unmixing,'' \emph{IEEE Journal of Selected Topics in Applied
  Earth Observations and Remote Sensing}, vol.~7, no.~6, pp. 1879--1888, June
  2014.

\bibitem{Wu2010}
X.~Wu, X.~Li, and L.~Zhao, ``A kernel spatial complexity-based nonlinear
  unmixing method of hyperspectral imagery,'' in \emph{Proc. LSMS/ICSEE}, 2010,
  pp. 451--458.

\bibitem{Li2012blind}
X.~Li, J.~Cui, and L.~Zhao, ``Blind nonlinear hyperspectral unmixing based on
  constrained kernel nonnegative matrix factorization,'' \emph{Signal, Image
  and Video Processing}, vol.~8, no.~8, pp. 1555--1567, 2012.

\bibitem{Chen-2013-ID321}
J.~Chen, C.~Richard, and P.~Honeine, ``Nonlinear unmixing of hyperspectral data
  based on a linear-mixture/nonlinear-fluctuation model,'' \emph{IEEE
  Transactions on Signal Processing}, vol.~61, pp. 480--492, Jan 2013.

\bibitem{chen2013nonlinear}
------, ``Nonlinear estimation of material abundances in hyperspectral images
  with $\ell_{1}$-norm spatial regularization,'' \emph{IEEE Transactions on
  Geoscience and Remote Sensing}, vol.~52, no.~5, pp. 2654--2665, May 2014.

\bibitem{Altmann-2013-ID311}
Y.~Altmann, N.~Dobigeon, S.~McLaughlin, and J.-Y. Tourneret, ``Nonlinear
  spectral unmixing of hyperspectral images using gaussian processes,''
  \emph{IEEE Transactions on Signal Processing}, vol.~61, pp. 2442--2453, May
  2013.

\bibitem{ammanouil2016nonlinear}
R.~Ammanouil, A.~Ferrari, C.~Richard, and S.~Mathieu, ``Nonlinear unmixing of
  hyperspectral data with vector-valued kernel functions,'' \emph{IEEE
  Transactions on Image Processing}, vol.~26, no.~1, pp. 340--354, 2017.

\bibitem{shi2014StatialInfoReview}
C.~Shi and L.~Wang, ``Incorporating spatial information in spectral unmixing:
  {A} review,'' \emph{Remote Sensing of Environment}, vol. 149, pp. 70--87,
  2014.

\bibitem{imbiriba2018ULTRA}
T.~Imbiriba, R.~A. Borsoi, and J.~C.~M. Bermudez, ``A low-rank tensor
  regularization strategy for hyperspectral unmixing,'' in \emph{2018 IEEE
  Statistical Signal Processing Workshop (SSP)}, 2018, pp. 373--377.

\bibitem{iordache2012sunsal_TV}
M.-D. Iordache, J.~M. Bioucas-Dias, and A.~Plaza, ``Total variation spatial
  regularization for sparse hyperspectral unmixing,'' \emph{IEEE Transactions
  on Geoscience and Remote Sensing}, vol.~50, no.~11, pp. 4484--4502, 2012.

\bibitem{feng2016adaptiveRegularizationParameterSparseHU}
R.~Feng, Y.~Zhong, and L.~Zhang, ``Adaptive spatial regularization sparse
  unmixing strategy based on joint {MAP} for hyperspectral remote sensing
  imagery,'' \emph{IEEE Journal of Selected Topics in Applied Earth
  Observations and Remote Sensing}, vol.~9, no.~12, pp. 5791--5805, 2016.

\bibitem{drumetz2016blindUnmixingELMM}
L.~Drumetz, M.-A. Veganzones, S.~Henrot, R.~Phlypo, J.~Chanussot, and
  C.~Jutten, ``Blind hyperspectral unmixing using an extended linear mixing
  model to address spectral variability,'' \emph{IEEE Transactions on Image
  Processing}, vol.~25, no.~8, pp. 3890--3905, 2016.

\bibitem{imbiriba2018ULTRA_V}
T.~{Imbiriba}, R.~A. {Borsoi}, and J.~C.~M. {Bermudez}, ``Low-rank tensor
  modeling for hyperspectral unmixing accounting for spectral variability,''
  \emph{IEEE Transactions on Geoscience and Remote Sensing}, vol.~58, no.~3,
  pp. 1833--1842, March 2020.

\bibitem{borsoi2019icassp}
R.~A. {Borsoi}, T.~{Imbiriba}, and J.~C. {Moreira Bermudez}, ``Improved
  hyperspectral unmixing with endmember variability parametrized using an
  interpolated scaling tensor,'' in \emph{2019 IEEE International Conference on
  Acoustics, Speech and Signal Processing (ICASSP)}, May 2019, pp. 2177--2181.

\bibitem{hong2019augmentedLMMvariability}
D.~Hong, N.~Yokoya, J.~Chanussot, and X.~X. Zhu, ``An augmented linear mixing
  model to address spectral variability for hyperspectral unmixing,''
  \emph{IEEE Transactions on Image Processing}, vol.~28, no.~4, pp. 1923--1938,
  2019.

\bibitem{hong2018SULoRA_lowRankEnbeddingUnmixingVar}
D.~Hong and X.~X. Zhu, ``Sulora: Subspace unmixing with low-rank attribute
  embedding for hyperspectral data analysis,'' \emph{IEEE Journal of Selected
  Topics in Signal Processing}, vol.~12, no.~6, pp. 1351--1363, 2018.

\bibitem{drumetz2017relationshipsBilinearELMM}
L.~Drumetz, B.~Ehsandoust, J.~Chanussot, B.~Rivet, M.~Babaie-Zadeh, and
  C.~Jutten, ``Relationships between nonlinear and space-variant linear models
  in hyperspectral image unmixing,'' \emph{IEEE Signal Processing Letters},
  vol.~24, no.~10, pp. 1567--1571, 2017.

\bibitem{tang2018spatialRegNonlinearUnmixing}
M.~Tang, L.~Gao, A.~Marinoni, P.~Gamba, and B.~Zhang, ``Integrating spatial
  information in the normalized p-linear algorithm for nonlinear hyperspectral
  unmixing,'' \emph{IEEE Journal of Selected Topics in Applied Earth
  Observations and Remote Sensing}, vol.~11, no.~4, pp. 1179--1190, 2018.

\bibitem{chen2014nonlinear2}
J.~Chen, C.~Richard, and P.~Honeine, ``Nonlinear estimation of material
  abundances in hyperspectral images with {L1}-norm spatial regularization,''
  \emph{Geoscience and Remote Sensing, IEEE Transactions on}, vol.~52, no.~5,
  pp. 2654--2665, 2014.

\bibitem{wang2017centralizedNonlocalSparseUnmixing}
R.~Wang, H.-C. Li, W.~Liao, X.~Huang, and W.~Philips, ``Centralized
  collaborative sparse unmixing for hyperspectral images,'' \emph{IEEE Journal
  of Selected Topics in Applied Earth Observations and Remote Sensing},
  vol.~10, no.~5, pp. 1949--1962, 2017.

\bibitem{yao2019nonlocalTV_NMF_unmixing}
J.~Yao, D.~Meng, Q.~Zhao, W.~Cao, and Z.~Xu, ``Nonconvex-sparsity and
  nonlocal-smoothness-based blind hyperspectral unmixing,'' \emph{IEEE
  Transactions on Image Processing}, vol.~28, no.~6, pp. 2991--3006, 2019.

\bibitem{lu2012manifoldRegularizedSparseNMF}
X.~Lu, H.~Wu, Y.~Yuan, P.~Yan, and X.~Li, ``Manifold regularized sparse {NMF}
  for hyperspectral unmixing,'' \emph{IEEE Transactions on Geoscience and
  Remote Sensing}, vol.~51, no.~5, pp. 2815--2826, 2012.

\bibitem{ammanouil2015graphLaplacianRegUnmixing}
R.~Ammanouil, A.~Ferrari, and C.~Richard, ``A graph laplacian regularization
  for hyperspectral data unmixing,'' in \emph{2015 IEEE International
  Conference on Acoustics, Speech and Signal Processing (ICASSP)}.\hskip 1em
  plus 0.5em minus 0.4em\relax IEEE, 2015, pp. 1637--1641.

\bibitem{Borsoi_multiscale_lgrs_2018}
R.~A. {Borsoi}, T.~{Imbiriba}, J.~C.~M. {Bermudez}, and C.~{Richard}, ``A fast
  multiscale spatial regularization for sparse hyperspectral unmixing,''
  \emph{IEEE Geoscience and Remote Sensing Letters}, vol.~16, no.~4, pp.
  598--602, April 2019.

\bibitem{Borsoi_multiscaleVar_2018}
R.~A. {Borsoi}, T.~{Imbiriba}, and J.~C.~M. {Bermudez}, ``A data dependent
  multiscale model for hyperspectral unmixing with spectral variability,''
  \emph{IEEE Transactions on Image Processing}, vol.~29, pp. 3638--3651, 2020.

\bibitem{song2016regularizationParamEstimHSdeconvolution}
Y.~Song, D.~Brie, E.-H. Djermoune, and S.~Henrot, ``Regularization parameter
  estimation for non-negative hyperspectral image deconvolution,'' \emph{IEEE
  Transactions on Image Processing}, vol.~25, no.~11, pp. 5316--5330, 2016.

\bibitem{Mercer1909}
J.~Mercer, ``Functions of positive and negative type and their connection with
  the theory of integral equations,'' \emph{Philos. Trans. Roy. Soc. London
  Ser. A}, vol. 209, pp. 415--446, 1909.

\bibitem{Moore1916}
E.~H. Moore, ``On properly positive hermitian matrices,'' \emph{Bull. American
  Mathematical Society}, vol.~23, p.~59, 1916.

\bibitem{Aronszajn1950}
N.~Aronszajn, ``Theory of reproducing kernels,'' \emph{Transactions of the
  American Mathematical Society}, vol.~68, 1950.

\bibitem{kreyszig1989introductory}
E.~Kreyszig, \emph{Introductory functional analysis with applications}.\hskip
  1em plus 0.5em minus 0.4em\relax Wiley New York, 1989, vol.~81.

\bibitem{steinwart2008support}
I.~Steinwart and A.~Christmann, \emph{Support vector machines}.\hskip 1em plus
  0.5em minus 0.4em\relax Springer, 2008.

\bibitem{Vapnik1995}
V.~N. Vapnik, \emph{The nature of statistical learning theory}.\hskip 1em plus
  0.5em minus 0.4em\relax New York, NY: Springer, 1995.

\bibitem{ScholkopfBook:2001}
B.~Sch\"{o}lkopf and A.~J. Smola, \emph{{Learning with Kernels: Support Vector
  Machines, Regularization, Optimization, and Beyond}}.\hskip 1em plus 0.5em
  minus 0.4em\relax The MIT Press, 2001.

\bibitem{Rasmussen-2006-ID292}
C.~E. Rasmussen and C.~K.~I. Williams, \emph{Gaussian Processes for Machine
  Learning}.\hskip 1em plus 0.5em minus 0.4em\relax The MIT Press, 2006.

\bibitem{Somers:2009p6577}
B.~Somers, K.~Cools, S.~Delalieux, J.~Stuckens, D.~V. der Zande, W.~W.
  Verstraeten, and P.~Coppin, ``Nonlinear hyperspectral mixture analysis for
  tree cover estimates in orchards,'' \emph{Remote Sensing of Environment},
  vol. 113, no.~6, pp. 1183--1193, February 2009.

\bibitem{Suykens2002}
J.~A.~K. Suykens, T.~V. Gestel, J.~D. Brabanter, B.~D. Moor, and J.~Vandewalle,
  \emph{Least Squares Support Vector Machines}.\hskip 1em plus 0.5em minus
  0.4em\relax Singapore: World Scientific, 2002.

\bibitem{shuman2013signalProcessingGraphsReview}
D.~I. Shuman, S.~K. Narang, P.~Frossard, A.~Ortega, and P.~Vandergheynst, ``The
  emerging field of signal processing on graphs: Extending high-dimensional
  data analysis to networks and other irregular domains,'' \emph{IEEE signal
  processing magazine}, vol.~30, no.~3, pp. 83--98, 2013.

\bibitem{stevens2017graphConstructionHSIs}
J.~R. Stevens, R.~G. Resmini, and D.~W. Messinger, ``Spectral-density-based
  graph construction techniques for hyperspectral image analysis,'' \emph{IEEE
  Transactions on Geoscience and Remote Sensing}, vol.~55, no.~10, pp.
  5966--5983, 2017.

\bibitem{achanta2012slicPAMI}
R.~Achanta, A.~Shaji, K.~Smith, A.~Lucchi, P.~Fua, and S.~S{\"u}sstrunk,
  ``{SLIC} superpixels compared to state-of-the-art superpixel methods,''
  \emph{IEEE Transactions on Pattern Analysis and Machine Intelligence},
  vol.~34, no.~11, pp. 2274--2282, 2012.

\bibitem{arbelaez2006ultrametricContourMaps}
P.~Arbelaez, ``Boundary extraction in natural images using ultrametric contour
  maps,'' in \emph{2006 Conference on Computer Vision and Pattern Recognition
  Workshop (CVPRW'06)}.\hskip 1em plus 0.5em minus 0.4em\relax IEEE, 2006, pp.
  182--182.

\bibitem{veganzones2014hyperspectralSegmentationBPT}
M.~A. Veganzones, G.~Tochon, M.~Dalla-Mura, A.~J. Plaza, and J.~Chanussot,
  ``Hyperspectral image segmentation using a new spectral unmixing-based binary
  partition tree representation,'' \emph{IEEE Transactions on Image
  Processing}, vol.~23, no.~8, pp. 3574--3589, 2014.

\bibitem{wang2017superpixelSparseNMF}
X.~Wang, Y.~Zhong, L.~Zhang, and Y.~Xu, ``Spatial group sparsity regularized
  nonnegative matrix factorization for hyperspectral unmixing,'' \emph{IEEE
  Transactions on Geoscience and Remote Sensing}, vol.~55, no.~11, pp.
  6287--6304, 2017.

\bibitem{golub1979generalizedCrossValidationParameterEst}
G.~H. Golub, M.~Heath, and G.~Wahba, ``Generalized cross-validation as a method
  for choosing a good ridge parameter,'' \emph{Technometrics}, vol.~21, no.~2,
  pp. 215--223, 1979.

\bibitem{deledalle2014parameterEstimationSUGAR}
C.-A. Deledalle, S.~Vaiter, J.~Fadili, and G.~Peyr{\'e}, ``{Stein Unbiased
  GrAdient estimator of the Risk (SUGAR) for multiple parameter selection},''
  \emph{SIAM Journal on Imaging Sciences}, vol.~7, no.~4, pp. 2448--2487, 2014.

\bibitem{ammanouil2018adapt_parameterEstimation}
R.~Ammanouil, A.~Ferrari, and C.~Richard, ``{ADA-PT}: An adaptive parameter
  tuning strategy based on the weighted stein unbiased risk estimator,'' in
  \emph{2018 IEEE International Conference on Acoustics, Speech and Signal
  Processing (ICASSP)}.\hskip 1em plus 0.5em minus 0.4em\relax IEEE, 2018, pp.
  4449--4453.

\bibitem{belge2002multipleParameterLcurve}
M.~Belge, M.~E. Kilmer, and E.~L. Miller, ``Efficient determination of multiple
  regularization parameters in a generalized {L-curve} framework,''
  \emph{Inverse Problems}, vol.~18, no.~4, pp. 1161--1183, 2002.

\bibitem{hall1987chiSquaredParameterEstimationGCV}
P.~Hall and D.~M. Titterington, ``Common structure of techniques for choosing
  smoothing parameters in regression problems,'' \emph{Journal of the Royal
  Statistical Society: Series B (Methodological)}, vol.~49, no.~2, pp.
  184--198, 1987.

\bibitem{thompson1991studyParameterSelectionRegularizationPAMI}
A.~M. Thompson, J.~C. Brown, J.~W. Kay, and D.~M. Titterington, ``A study of
  methods of choosing the smoothing parameter in image restoration by
  regularization,'' \emph{IEEE Transactions on Pattern Analysis \& Machine
  Intelligence}, no.~4, pp. 326--339, 1991.

\bibitem{galatsanos1992regularizationParamChoiceImageRestoration}
N.~P. Galatsanos and A.~K. Katsaggelos, ``Methods for choosing the
  regularization parameter and estimating the noise variance in image
  restoration and their relation,'' \emph{IEEE Transactions on image
  processing}, vol.~1, no.~3, pp. 322--336, 1992.

\bibitem{engl1996regularizationInverseProbs}
A.~N. Heinz Werner~Engl, Martin~Hanke, \emph{Regularization of Inverse Problems
  (Mathematics and Its Applications)}, 1st~ed., ser. Mathematics and Its
  Applications.\hskip 1em plus 0.5em minus 0.4em\relax Springer, 1996.

\bibitem{hunt1973parameterEstimationChiSquaredEqualityConstraint}
B.~R. Hunt, ``The application of constrained least squares estimation to image
  restoration by digital computer,'' \emph{IEEE Transactions on Computers},
  vol. 100, no.~9, pp. 805--812, 1973.

\bibitem{altmann2015spatialBayesianNonlinearUnmixing}
Y.~Altmann, M.~Pereyra, and S.~McLaughlin, ``Bayesian nonlinear hyperspectral
  unmixing with spatial residual component analysis,'' \emph{IEEE Transactions
  on Computational Imaging}, vol.~1, no.~3, pp. 174--185, 2015.

\bibitem{ammanouil2016spatialNonlinearUnmixing}
R.~Ammanouil, A.~Ferrari, C.~Richard, and J.-Y. Tournere, ``Spatial
  regularization for nonlinear unmixing of hyperspectral data with
  vector-valued kernel functions,'' in \emph{2016 IEEE Statistical Signal
  Processing Workshop (SSP)}.\hskip 1em plus 0.5em minus 0.4em\relax IEEE,
  2016, pp. 1--5.

\bibitem{boyd2004cvxbook}
S.~Boyd and L.~Vandenberghe, \emph{Convex optimization}.\hskip 1em plus 0.5em
  minus 0.4em\relax Cambridge university press, 2004.

\bibitem{tuy2013strongDualityQPQC}
H.~Tuy and H.~D. Tuan, ``Generalized {S-lemma} and strong duality in nonconvex
  quadratic programming,'' \emph{Journal of Global Optimization}, vol.~56,
  no.~3, pp. 1045--1072, 2013.

\bibitem{borsoi2019BMUAN_arxiv}
R.~A. Borsoi, T.~Imbiriba, J.~C.~M. Bermudez, and C.~Richard, ``A blind
  multiscale spatial regularization framework for kernel-based spectral
  unmixing,'' \emph{arXiv preprint arXiv:1908.06925}, 2019.

\bibitem{galvan2017multivariateBisection}
M.~L. Galv{\'{a}}n, ``The multivariate bisection algorithm,'' \emph{Revista de
  la Uni{\'{o}}n Matem{\'{a}}tica Argentina}, pp. 79--98, Mar. 2019.

\bibitem{bachrathy2012multidim_bisection}
D.~Bachrathy and G.~St{\'e}p{\'a}n, ``Bisection method in higher dimensions and
  the efficiency number,'' \emph{Periodica Polytechnica Mechanical
  Engineering}, vol.~56, no.~2, pp. 81--86, 2012.

\bibitem{yi2018superpixelHomogeneityTesting}
J.~Yi and M.~Velez-Reyes, ``Low-dimensional enhanced superpixel representation
  with homogeneity testing for unmixing of hyperspectral imagery,'' in
  \emph{Algorithms and Technologies for Multispectral, Hyperspectral, and
  Ultraspectral Imagery XXIV}, vol. 10644.\hskip 1em plus 0.5em minus
  0.4em\relax International Society for Optics and Photonics, 2018, p. 1064422.

\bibitem{halimi2016unmixingVariabilityNonlinearityMismodeling}
A.~Halimi, P.~Honeine, and J.~M. Bioucas-Dias, ``Hyperspectral unmixing in
  presence of endmember variability, nonlinearity, or mismodeling effects,''
  \emph{IEEE Transactions on Image Processing}, vol.~25, no.~10, pp.
  4565--4579, 2016.

\bibitem{roger1996residualMethodNoiseCovarEstimation}
R.~E. Roger, ``Principal components transform with simple, automatic noise
  adjustment,'' \emph{International Journal of Remote Sensing}, vol.~17,
  no.~14, pp. 2719--2727, 1996.

\bibitem{mahmood2017modifiedNoiseCovarianceEstimation}
A.~Mahmood, A.~Robin, and M.~Sears, ``Modified residual method for the
  estimation of noise in hyperspectral images.'' \emph{IEEE Transactions on
  Geoscience and Remote Sensing}, vol.~55, no.~3, pp. 1451--1460, 2017.

\bibitem{Nascimento2005}
J.~M.~P. Nascimento and J.~M. Bioucas-Dias, ``{Vertex Component Analysis}: A
  fast algorithm to unmix hyperspectral data,'' \emph{IEEE Transactions on
  Geoscience and Remote Sensing}, vol.~43, no.~4, pp. 898--910, April 2005.

\bibitem{ai2009strongDualityQPQC_SDP}
W.~Ai and S.~Zhang, ``Strong duality for the {CDT} subproblem: a necessary and
  sufficient condition,'' \emph{SIAM Journal on Optimization}, vol.~19, no.~4,
  pp. 1735--1756, 2009.

\bibitem{bomze2015copositiveRelaxationBeatsLagrangian}
I.~M. Bomze, ``Copositive relaxation beats {Lagrangian} dual bounds in
  quadratically and linearly constrained quadratic optimization problems,''
  \emph{SIAM Journal on Optimization}, vol.~25, no.~3, pp. 1249--1275, 2015.

\end{thebibliography}

\vspace{-1.15cm}
\begin{IEEEbiographynophoto}{Ricardo Augusto Borsoi (S'18)} 
received the MSc degree in electrical engineering from Federal University of Santa Catarina (UFSC), Florian\'opolis, Brazil, in 2016. He is currently working towards his doctoral degree at Universit\'e C\^ote d'Azur (OCA) and at UFSC. His research interests include image processing, tensor decomposition, and hyperspectral image analysis.
\end{IEEEbiographynophoto}

\vspace{-1.15cm}
\begin{IEEEbiographynophoto}{Tales Imbiriba (S'14, M'17)}   
received his Doctorate degree from the Department of Electrical Engineering (DEE) of the Federal University of Santa Catarina (UFSC), Florian\'opolis, Brazil, in 2016. He served as a Postdoctoral Researcher (2017--2019) at the DEE--UFSC and is currently a Postdoctoral Researcher at the ECE dept. of the Northeastern University, Boston, MA, USA. 
His research interests include audio and image processing, pattern recognition, kernel methods, adaptive filtering, and Bayesian Inference.
\end{IEEEbiographynophoto}

\vspace{-1.15cm}
\begin{IEEEbiographynophoto}{Jos\'e Carlos M. Bermudez (S'78,M'85,SM'02)}
  received the B.E.E. degree from the Federal University of Rio de Janeiro (UFRJ), Rio de Janeiro, Brazil, the M.Sc. degree from COPPE/UFRJ, and the Ph.D. degree from Concordia University, Montreal, Canada, in 1978, 1981, and 1985, respectively.
  He is a Professor at UFSC and at Catholic University of Pelotas (UCPel), Pelotas, Brazil. He has held the position of Visiting Researcher several times for periods of one month at INPT, Toulouse, France, and Université Nice Sophia-Antipolis, France. He spent sabbatical years at the University of California, Irvine (UCI), USA, in 1994, and at INPT, France, in 2012.
  His research interests are in statistical signal processing, including linear and nonlinear adaptive filtering, image processing, hyperspectral image processing and machine learning. Prof. Bermudez served as an Associate Editor of the IEEE TRANSACTIONS ON SIGNAL PROCESSING from 1994 to 1996 and from 1999 to 2001, as an Associate Editor of the EURASIP Journal of Advances on Signal Processing from 2006 to 2010, and as a Senior Area Editor of the IEEE TRANSACTIONS ON SIGNAL PROCESSING from 2015 to 2019.  He is the Chair of the Signal Processing Theory and Methods Technical Committee of the IEEE Signal Processing Society (2019-2020). Prof. Bermudez is a \mbox{Senior Member of the IEEE.}
\end{IEEEbiographynophoto}

\vspace{-1.15cm}
\begin{IEEEbiographynophoto}{Cédric Richard (Senior Member, IEEE)}
received the Dipl.-Ing. and master's and the Ph.D. degree in electrical and computer engineering from Compiègne University of Technology, Compiègne, France, in 1994 and 1998, respectively. He is a Full Professor with the Université Côte d’Azur, Nice, France. During 2010--2015, he was distinguished as a Member of the Institut Universitaire de France. His current research interests include statistical signal processing and machine learning. He is the author of more than 300 papers. Prof. Richard is the Director-at-Large of Region 8 (Europe, Middle East, and Africa) of the IEEE Signal Processing Society (IEEE-SPS) and a Member of the Board of Governors of the IEEE-SPS. He is also the Director of the French Federal CNRS research association Information, Signal, Image et Vision. Since 2019, he has been an Associate Editor for the IEEE OPEN JOURNAL OF SIGNAL PROCESSING, and since 2009, as an Associate Editor for Signal Processing (Elsevier). During 2015--2018, he was a Senior Area Editor for the IEEE TRANSACTIONS ON SIGNAL PROCESSING and an Associate Editor for the IEEE TRANSACTIONS ON SIGNAL AND INFORMATION PROCESSING OVER NETWORKS. He was also an Associate Editor for the IEEE TRANSACTIONS ON SIGNAL PROCESSING during 2006--2010. He is an elected member of the IEEE Signal Processing Theory and Methods Technical Committee during 2009--2014 and 2018--present, and was an elected member of the IEEE Machine Learning for Signal Processing Technical Committee during 2012--2018. 
\end{IEEEbiographynophoto}


\clearpage
\appendices
\section{Supplemental Material: Proof of Theorem~\ref{thm:thm1_duality_ours}}
\label{sec:appendix:strong_duality_proof}

For the sake of brevity, we provide a demonstration for Theorem~\ref{thm:thm1_duality_ours} only for the problem~\eqref{eq:khype_opt_finescale_constr_3}. The demonstration for problem~\eqref{eq:khype_opt_coarse_constr} follows the same ideas and is thus straightforward.
The first part of the proof consists in using the same ideas as in the representer theorem (presented in Section~\ref{sec:kernel}) to obtain a finite dimensional representation of~\eqref{eq:khype_opt_finescale_constr_3} in its primal form. Note that, since $\cp{H}$ is formed by a class of functions defined as linear combinations of $\kappa(\cdot, \widetilde{\bm})$, $\widetilde{\bm}\in \cp{M}$, we can write any candidate $\psi_n\in\cp{H}$, $n=1,\ldots,N$ for solving~\eqref{eq:khype_opt_finescale_constr_3} equivalently 
as
\begin{align}
    \psi_n = \sum_{\ell=1}^L \beta_{n,\ell} \kappa(\cdot,\widetilde{\bm}_{\ell}) + \psi_{n}^\perp \,,
\end{align}
where $\psi_{n}^\perp$ is orthogonal to $\kappa(\cdot,\widetilde{\bm}_{\ell})$, $\ell=1,\ldots,L$, i.e., $\langle \kappa(\cdot,\widetilde{\bm}_{\ell}), \psi_{n}^\perp\rangle_{\mathcal{H}}=0$, and $\beta_{n,\ell}$ are the linear combination coefficients. This allows us to write each position of the vector $\psi_n(\bM)$ as
\begin{align}
    \big[\psi_{n}(\bM)\big]_j & = \psi_{n}(\widetilde{\bm}_{j}) 
    \nonumber\\
    & = \langle\psi_{n},\kappa(\cdot,\widetilde{\bm}_j)\rangle_{\mathcal{H}}
    \nonumber\\
    & = \Big\langle\sum_{\ell=1}^L\beta_{n,\ell}\kappa(\cdot,\widetilde{\bm}_{\ell})+\psi_{n}^\perp,\kappa(\cdot,\widetilde{\bm}_j) \Big\rangle_{\mathcal{H}} 
    \nonumber\\
    & = \sum_{\ell=1}^L\beta_{n,\ell}\langle\kappa(\cdot,\widetilde{\bm}_{\ell}),\kappa(\cdot,\widetilde{\bm}_j)\rangle_{\mathcal{H}}  \,,
\end{align}
for $j=1,\ldots,L$. Thus,
\begin{align}
    \psi_{n}(\bM) = \bK \bbeta_n \,.
\end{align}
Note that $\widehat{\psi}_{\mathcal{C}_n}(\bM)$ in~\eqref{eq:khype_opt_finescale_constr_3} is a constant vector since it was estimated in the previous problem. The objective function can be similarly written as
\begin{align} \label{eq:ObjectiveFunction}
    & \mathop{\arg\min}_{\{\ba_{n},\psi_{n}\}} \,\,\, 
    \frac{1}{2} \sum_{n=1}^N \|\psi_{n}\|_{\mathcal{H}}^2
    \nonumber \\
    ={} & \mathop{\arg\min}_{\{\ba_{n},\psi_{n}\}} \,\,\, 
    \frac{1}{2} \sum_{n=1}^N \bbeta_n^\top\bK\bbeta_n + \frac{1}{2} \sum_{n=1}^N \|\psi_{n}^\perp\|_{\mathcal{H}}^2 \,.
\end{align}
Since the constraints do not depend on $\psi_{n}^\perp$, the second term of \eqref{eq:ObjectiveFunction} is irrelevant to the problem, and will be equal to zero for any solution to the optimization problem. Thus, we can rewrite~\eqref{eq:khype_opt_finescale_constr_3} as a finite dimensional optimization problem
\begin{align} \label{eq:finescale_prob_appdx_ii}
    & \mathop{\arg\min}_{\{\ba_{n},\bbeta_{n}\}} \,\,\, 
    \frac{1}{2} \sum_{n=1}^N \bbeta_n^\top\bK\bbeta_n
    \\
    & \text{subject to } \ba_{n} \geq \cb{0} \,, \,\, \cb{1}^\top\ba_{n} = 1 \,,
    \,\,\,\, n=1,\ldots,N
    \nonumber \\ & \hspace{1.5cm}
    \bxi_{n} {}={} \by_{n} - \bM \ba_{n} - \bK\bbeta_n \,,
    \,\,\,\, n=1,\ldots,N
    \nonumber \\
    & \hspace{1.5cm} \frac{1}{2}\sum_{n=1}^N\|\bxi_n\|_2^2 
    {}={} \frac{N}{2} C_1 \,,
    \nonumber \\
    & \hspace{1.5cm}  \frac{1}{2}\sum_{n=1}^N
    \big(\|\ba_n-\widehat{\ba}_{\mathcal{D}_n}\|_2^2 + \|\bxi_{\psi,n}\|^2 \big)
	{}={} \! \frac{N}{2} (C_Y - C_E) \,,
    \nonumber \\ \nonumber 
    & \hspace{1.5cm}  \bxi_{\psi,n} = \bM^{\dagger}\big(\bK\bbeta_n  - \widehat{\psi}_{\mathcal{C}_n}(\bM)\big) \,,
    \,\,\, n=1,\ldots,N \,.
\end{align}
This derivation was just to show that we can write the primal version of the optimization problem in finite dimension. The finite dimensional Lagrangian is given by
\begin{align} \label{eq:finescale_prob_appdx_ii_lagrangian}
    \widetilde{\mathcal{J}}_{\calD} {}={} \! & \sum_{n=1}^N \bigg\{
    \frac{1}{2} \bbeta_n^\top\bK\bbeta_n
    + \frac{\mu_1}{2}\Big(\|\bxi_n\|_2^2 - C_1\Big)
    + \lambda_n(\cb{1}^\top\ba_{n} - 1)
    \nonumber \\ &
    + \frac{\mu_2}{2} \Big(
    \|\ba_n-\widehat{\ba}_{\mathcal{D}_n}\|_2^2 - C_Y
    + \|\bxi_{\psi,n}\|^2 
    + C_E \Big)
    \nonumber \\ &
    + \bmu_{3,n}^\top \Big(\bM^{\dagger}(\bK\bbeta_n  - \widehat{\psi}_{\mathcal{C}_n}(\bM)\big) - \bxi_{\psi,n} \Big)
    \nonumber \\ &
    - \bbeta_{n}^\top \big(\bxi_{n}-\by_{n} + \bM\ba_n + \bK\bbeta_n \big)
    %
    - \bgamma_n^\top\ba_n
    \bigg\} \,.
\end{align}

Now, we can write~\eqref{eq:finescale_prob_appdx_ii} equivalently by splitting each quadratic equality constraints in two quadratic inequality ones, one being convex and the other non-convex. This results in the following problem:
\begin{align} \label{eq:finescale_prob_appdx_iii}
    & \min_{\bx\in \mathscr{A}} \,\,\, f(\bx)
    \\
    & \text{subject to } g_1(\bx) \leq 0 \,, \quad g_2(\bx) \leq 0 \,,
    \nonumber\\
    & \hspace{1.5cm} g_3(\bx) \leq 0 \,, \quad g_4(\bx) \leq 0 \,,
    \nonumber\\\nonumber
    & \hspace{1.5cm} h_{n,p}(\bx) \leq 0 \,, \, n=1,\ldots,N \,, \, p=1,\ldots,P \,,
\end{align}
where $\bx=[\bx_1^\top,\ldots,\bx_N^\top]^\top$, with $\bx_n=[\ba_n^\top,\bbeta_n^\top,\bxi_n^\top,\bxi_{\psi,n}^\top]^\top$, $\mathscr{A}$ is an affine manifold that encapsulates the equality constraints, functions $f(\bx)$, $g_i(\bx)$ and $h_{i,j}(\bx)$, for all $i$ and $j$ are quadratic and defined as
\begin{align}
    f(\bx) \equiv{} & \frac{1}{2} \sum_{n=1}^N \bbeta_n^\top\bK\bbeta_n \,,
    \\
    g_1(\bx) \equiv{} & \frac{1}{2}\sum_{n=1}^N\|\bxi_n\|_2^2 - \frac{N}{2} C_1 \,,
    \\
    g_2(\bx) \equiv{} & \frac{1}{2}\sum_{n=1}^N  \big(\|\ba_n-\widehat{\ba}_{\mathcal{D}_n}\|_2^2 + \|\bxi_{\psi,n}\|^2 \big) 
    \nonumber \\ &  - \frac{N}{2} (C_Y - C_E) \,,
    \\
    g_3(\bx) ={} & - g_1(\bx) \,,
    \\
    g_4(\bx) ={} & - g_2(\bx) \,,
    \\
    h_{n,p}(\bx) \equiv{} & -a_{n,p} \,.
\end{align}
Note that the constraints involving $g_3(\bx)$ and $g_4(\bx)$ are nonconvex, while the remaining ones are all convex. Note also that this problem does not satisfy the Slater condition, what imposes difficulties in the analysis and precludes the consideration of works that rely on this hypothesis (e.g.,~\cite{ai2009strongDualityQPQC_SDP,bomze2015copositiveRelaxationBeatsLagrangian}).
The Lagrangian of this problem is 
\begin{align} \label{eq:finescale_prob_appdx_iii_lagrangian_i}
    \mathcal{L}_1(\bx,{\bu},\bzeta) ={} & f(\bx) + \sum_{n,p} u_{(n-1)P+n+4} \,h_{n,p}(\bx)
    \nonumber \\[-0.2cm] &  + \sum_{j=1}^4 u_j\, g_j(\bx)  + L_{\mathscr{A}}(\bx,\bzeta) \,,
\end{align}
where $\bu=[u_{1},\ldots,u_{NP+4}]^\top$, with $u_{j}\in\mathbb{R}_+$ and $\bzeta\in\mathbb{R}_+^{(L+P+1)N}$ are the Lagrange multipliers, and $L_{\mathscr{A}}(\bx,\bzeta)$ are the terms in the Lagrangian associated with the linear equality constraints.

Since the constraints involving $g_1(\bx)$ and $g_3(\bx)$, and $g_2(\bx)$ and $g_4(\bx)$ are linearly dependent, we have
\begin{align}
    u_1\,g_1(\bx)+u_3\,g_3(\bx) & = (u_1-u_3) g_1(\bx) \,,
    \\
    u_2\,g_2(\bx)+u_4\,g_4(\bx) & = (u_2-u_4) g_2(\bx) \,.
\end{align}
Thus, we can define new variables $v_1=(u_1-u_3)\in\mathbb{R}$ and $v_2=(u_2-u_4)\in\mathbb{R}$ and rewrite~\eqref{eq:finescale_prob_appdx_iii_lagrangian_i} as
\begin{align} \label{eq:finescale_prob_appdx_iii_lagrangian_ii}
    \mathcal{L}_2(\bx,\widetilde{\bu},\bzeta,v_1,v_2) ={} & f(\bx) + \sum_{n,p} u_{(n-1)P+n+4} + v_1\,g_1(\bx)
    \nonumber \\ & + v_1\,g_2(\bx) + L_{\mathscr{A}}(\bx,\bzeta) \,,
\end{align}
where $\widetilde{\bu}=[u_{5},\ldots,u_{NP+4}]^\top$. Since we can always find nonnegative $u_j$, $j=1,\ldots,4$ satisfying ${v}_1=(u_1-u_3)$ and ${v}_2=(u_2-u_4)$, the optimization of the Lagrangian in~\eqref{eq:finescale_prob_appdx_iii_lagrangian_ii} is equivalent to the optimization of the original one in~\eqref{eq:finescale_prob_appdx_ii_lagrangian} (subject to the corresponding nonnegativity constraints). Thus, this means that strong duality of problem~\eqref{eq:finescale_prob_appdx_iii} implies that problem~\eqref{eq:finescale_prob_appdx_ii} has strong duality too.

In order to study strong duality of~\eqref{eq:finescale_prob_appdx_iii}, we resort to the following theorem, originally stated in~\cite{tuy2013strongDualityQPQC}.
\begin{theorem}{\cite[Theorem 6]{tuy2013strongDualityQPQC}}
\label{thm:duality_theorem_TUY}
In~\eqref{eq:finescale_prob_appdx_iii}, let 
\begin{align}
    L(\bx,\bu) = f(\bx) + \sum_{j=1}^4 u_j\,g_j(\bx) + \sum_{n,p} u_{(n-1)P+n+4} \,h_{n,p}(\bx)
    \nonumber
\end{align}
and assume that the concave function $l(\bu):\bu\mapsto\inf_{\bx\in\mathscr{A}}L(\bx,\bu)$ attains its maximum at a point $\bu^*\in\mathbb{R}_+^{NP+4}$ such that $L(\bx,\bu^*)$ is strictly convex on the set $\mathscr{A}$. More precisely, 
\begin{align}
    \nonumber
    \exists \bu^*\in\mathop{\arg\max}_{\bu\in\mathbb{R}_+^{NP+4}} \big\{\inf_{\bx\in\mathscr{A}}L(\bx,\bu)\big\} : \Hess \big[L(\bx,\bu^*)\big]\succ\cb{0} \,,
\end{align}
where $\Hess(\cdot)$ is the Hessian operator. Then, strong \mbox{duality holds.}
\end{theorem}

In order to apply this to our problem, we note that since $h_{n,p}(\bx)$ are linear functions, it is sufficient that $u_1-u_3>0$ and $u_2-u_4>0$ for ${L}(\bx,\bu)$ to be strictly convex. Moreover, under the change of variables $\mu_1=(u_1-u_3)\in\mathbb{R}$ and $\mu_2=(u_2-u_4)\in\mathbb{R}$, maximization of $l(\bu)$ in Theorem~\eqref{thm:duality_theorem_TUY} becomes equivalent to the dual problem~\eqref{eq:dual_problem2_i} (due to the equivalence we showed between~\eqref{eq:lagrangian_finescale_i} and~\eqref{eq:finescale_prob_appdx_iii_lagrangian_ii}). 
Since we assumed that the optimal solutions $\mu_1^*$ and $\mu_2^*$ to~\eqref{eq:dual_problem2_i} are strictly positive, this means that $u_1^*-u_3^*>0$ and $u_2^*-u_4^*>0$ and thus ${L}(\bx,\bu^*)$ is strictly convex. Therefore, strong duality holds.


\section{Supplemental material: Alternative image reconstruction error metrics}

In addition to the RMSE, we considered three other quality measures typically employed with hyperspectral images to evaluate reconstructions of real data sets: the spectral angle mapper (SAM), the spectral information divergence (SID)~[S1], and a spectral-spatial similarity measure that considers the spatial neighborhood of each pixel, named IPD~[S2]. The results for the real datasets are shown in Table~\ref{tab:otherIMreconstructionMetrics}. These results do not lead to a significantly different qualitative conclusion regarding the best solution.

\begin{itemize}
    \item[{[S1]}] C.-I. Chang, ``Spectral  information  divergence  for  hyperspectral  image  analysis,''  in \textit{Proc. of the IEEE International Geoscience and Remote Sensing Symposium}, 1999, pp. 509--511.
    
    \item[{[S2]}] H. Pu, Z. Chen, B. Wang, and G.-M. Jiang, ``A novel spatial–spectral similarity measure for dimensionality reduction and classification of hyperspectral imagery,'' \textit{IEEE Transactions on Geoscience and Remote Sensing}, vol. 52, no. 11, pp. 7008--7022, 2014.
\end{itemize}

\begin{table}[]
\centering
\renewcommand{\arraystretch}{1.2}
\caption{Different metrics for the image reconstruction on real data. Best results are bold red, second best are bold blue.}
\label{tab:otherIMreconstructionMetrics}
\begin{tabular}{c|cccccccccccccccccccc}
\hline
& \multicolumn{4}{c}{Cuprite}\\
\hline
	&	RMSE	&	SAM	&	SID	&	IDP	\\
\hline
FCLS	&	0.0107	&	0.0285	&	0.00051	&	0.1347	\\
K-Hype	&	\cblue{\bf 0.0082}	&	\cblue{\bf 0.0227}	&	\cblue{\bf 0.00029}	&	\cblue{\bf 0.1085}	\\
K-Hype-TV	&	\cblue{\bf 0.0082}	&	0.0228	&	\cblue{\bf 0.00029}	&	0.1090	\\
CDA-NL	&	0.0095	&	0.0259	&	\cred{\bf 0.00031}	&	0.1284	\\
NDU	&	\cred{\bf 0.0080}	&	\cred{\bf 0.0220}	&	\cred{\bf 0.00031}	&	\cred{\bf 0.1035}	\\
BMUA-N	&	0.0090	&	0.0245	&	0.00035	&	0.1184	\\
\hline
& \multicolumn{4}{c}{Jasper Ridge}			\\					
\hline
	&	RMSE	&	SAM	&	SID	&	IDP	\\
\hline
FCLS	&	0.0225	&	0.0681	&	0.00261	&	0.2677	\\
K-Hype	&	0.0562	&	\cblue{\bf 0.0603}	&	0.00254	&	0.5859	\\
K-Hype-TV	&	\cred{\bf 0.0164}	&	0.0604	&	0.00240	&	\cred{\bf 0.2011}	\\
CDA-NL	&	0.0168	&	0.0643	&	\cblue{\bf 0.00228}	&	0.2040	\\
NDU	&	\cblue{\bf 0.0167}	&	\cred{\bf 0.0551}	&	\cred{\bf 0.00163}	&	\cblue{\bf 0.2038}	\\
BMUA-N	&	0.0170	&	0.0624	&	0.00235	&	0.2071	\\
\hline
& \multicolumn{4}{c}{Urban} \\		
\hline
	&	RMSE	&	SAM	&	SID	&	IDP	\\
\hline
FCLS	&	0.0228	&	0.0564	&	0.00081	&	0.2098	\\
K-Hype	&	0.0434	&	0.0346	&	0.00030	&	0.3670	\\
K-Hype-TV	&	\cred{\bf 0.0060}	&	\cblue{\bf 0.0339}	&	\cblue{\bf 0.00029}	&	\cred{\bf 0.0718}	\\
CDA-NL	&	0.0063	&	\cred{\bf 0.0331}	&	\cred{\bf 0.00028}	&	0.0744	\\
NDU	&	0.0159	&	0.0447	&	0.00051	&	0.1540	\\
BMUA-N	&	\cblue{\bf 0.0062}	&	0.0343	&	0.00030	&	\cblue{\bf 0.0739}	\\
\hline
\end{tabular}
\end{table}



\end{document}